\documentclass[runningheads,draft]{svjour2}

\smartqed  % flush right qed marks, e.g. at end of proof

\usepackage{amsmath,amssymb,gastex,array,latexsym}
\usepackage[latin1]{inputenc}
\usepackage{graphicx}

\newcommand{\era}[2]{\overset{#1}{\underset{#2}{\longrightarrow}}}
\newcommand{\erb}[1]{\overset{#1}{\longrightarrow}}

\newcommand{\eRa}[2]{\overset{#1}{\underset{#2}{\Longrightarrow}}}

\newcommand{\diese}{\textrm{\small \#}}

\newcommand{\eps}{\varepsilon} \newcommand{\sep}{\;|\;}

\newcommand{\intro}[1]{\emph{#1}}

\journalname{Acta Informatica}

\begin{document}

\title{Linearly Bounded Infinite Graphs\thanks{A preliminary version
    of this article appeared in MFCS 2005.}}

\author{Arnaud Carayol \and Antoine Meyer}

\institute{Arnaud Carayol \at {\sc Irisa}, Campus de Beaulieu, 35042
  Rennes Cedex, France  \\
  Tel.: +33299847278,  Fax: +33299847171\\
  \email{Arnaud.Carayol@irisa.fr} % \\
  \and Antoine Meyer \at {\sc Liafa}, Universit\'e Denis Diderot,
  Case 7014, 2 place Jussieu, 75251 Paris Cedex 05, France\\
  \email{Antoine.Meyer@liafa.jussieu.fr} }

\date{Received: date / Revised: date}

\maketitle

\begin{abstract}
  Linearly bounded Turing machines have been mainly studied as
  acceptors for context-sensitive languages. We define a natural class
  of infinite automata representing their observable computational
  behavior, called linearly bounded graphs.  These automata naturally
  accept the same languages as the linearly bounded machines defining
  them. We present some of their structural properties as well as
  alternative characterizations in terms of rewriting systems and
  context-sensitive transductions. Finally, we compare these graphs to
  rational graphs, which are another class of automata accepting the
  context-sensitive languages, and prove that in the bounded-degree
  case, rational graphs are a strict sub-class of linearly bounded
  graphs.
\end{abstract}

\section{Introduction}
\label{sect:intro}

One of the cornerstones of formal language theory is the hierarchy of
languages introduced by Chomsky in \cite{Chomsky59}. It rests on the
definition of four increasingly restricted classes of grammars, which
respectively generate the \intro{recursively enumerable},
\intro{context-sensitive}, \intro{context-free} and \intro{rational}
languages. All these classes were extensively studied, and have been
given several alternative characterizations using different kinds of
formalisms (or \emph{acceptors}). For instance, pushdown systems
characterize context-free languages, and linearly bounded Turing
machines (LBMs) characterize context-sensitive languages. More
recently, several authors have related these four classes of languages
to classes of infinite graphs (for a survey, see for instance
\cite{Thomas01}). Given a fixed initial vertex and a set of final
vertices, one can associate a language to a graph by considering the
set of all words labeling a path between the initial vertex and one of
the final vertices. In \cite{Caucal02}, a summary of four classes of
graphs accepting the four classes of the Chomsky hierarchy is
presented. They are the Turing graphs \cite{Caucal03tm}, rational
graphs \cite{Morvan00,Morvan01a,Rispal02}, prefix-recognizable graphs
\cite{Caucal96,Caucal03pr}, and finite graphs.

\smallskip

Several approaches exist to define classes of infinite graphs, among
which we will cite three. The first one is to consider the finite
acceptor of a formal language, and to build a graph representing the
structure of its computations: vertices represent configurations, and
each edge reflects the observable effect of an input on the
configuration.  One speaks of the \emph{transition graph} of the
acceptor. An interesting consequence is that the language of the graph
can be deduced from the language of the acceptor it was built from. A
second method proposed in \cite{Caucal02} is to consider the
Cayley-type graphs of some classes of word rewriting systems. Each
vertex is a normal form for a given rewriting system, and an edge
between two vertices represents the addition of a letter and
re-normalization by the rewriting system. Finally, a third possibility
is to directly define the edge relations in a graph using automata or
other formalisms.  One may speak of derivation, transduction or
computation graphs. In this approach, a path no longer represents a
run of an acceptor, but rather a composition of binary relations.

Both prefix-recognizable graphs and Turing graphs have alternative
definitions along all three approaches. Prefix-recognizable graphs are
defined as the graphs of recognizable prefix relations. In
\cite{Stirling00}, Stirling characterized them as the transition graphs of
pushdown systems. It was also proved that they coincide with the
Cayley-type graphs of prefix rewriting systems. As for Turing graphs,
Caucal showed that they can be seen indifferently as the transition
and computation graphs of Turing machines \cite{Caucal03tm}. They are
also the Cayley-type graphs of unrestricted rewriting systems.
Rational graphs, however, are only defined as transduction graphs
(using rational transducers) and as the Cayley-type graphs of
left-overlapping rewriting systems, and lack a characterization as
transition graphs.  In this paper, we are interested in defining
a suitable notion of transition graphs of linearly bounded Turing
machines, and to determine some of their structural properties as well
as to compare them with rational graphs.

As in \cite{Caucal03tm} for Turing machines, we first define a labeled
version of LBMs, called LLBMs. Their transition rules are labeled
either by a symbol from the input alphabet or by a special symbol
denoting an internal, unobservable transition. Following an idea from
\cite{Stirling00}, we consider that in every configuration of a LLBM,
either internal actions or inputs are allowed, but not both at a time.
This way, we can distinguish between internal and external
configurations. The transition graph of a LLBM is the graph whose
vertices are external configurations, and whose edges represent an
input followed by any finite number of silent transitions. This
definition is purely structural and associates a unique graph to any
given LLBM. For convenience, we call such graphs \emph{linearly
  bounded graphs}. A similar work was proposed in
\cite{Knapik99,Payet00}, where the class of configuration graphs of
LBMs up to weak bisimulation is studied. However, it provides no
formal definition associating LBMs to a class of \emph{real-time}
graphs (without edges labeled by silent transitions) representing
their observable computations.

To further illustrate the suitability of our notion, we provide two
alternative definitions of linearly bounded graphs. First, we prove
that they are isomorphic to the Cayley-type graphs
of length-decreasing rewriting systems. The second alternative
definition directly represents the edge relations of a linearly
bounded graph as a certain kind of context-sensitive transductions.
This allows us to straightforwardly deduce structural properties of
linearly bounded graphs, like their closure under synchronized product
(which was already known from \cite{Knapik99}) and under restriction
to a context-sensitive set of vertices. To conclude this study, we
show that linearly bounded graphs and rational graphs form incomparable
classes, even in the finite degree case. However, bounded degree
rational graphs are a strict sub-class of linearly bounded graphs.

\section{Preliminary Definitions}
\label{sect:defs}

A labeled, directed and simple \intro{graph} is a set $G \subseteq V
\times \Sigma \times V$ where $\Sigma$ is a finite set of labels and
$V$ a countable set of \intro{vertices}. An element $(s,a,t)$ of $G$
is an \intro{edge} of \intro{source} $s$, \intro{target} $t$ and
\intro{label} $a$, and is written $s \era{a}{G} t$ or simply $s
\erb{a} t$ if $G$ is understood. The set of all sources and targets of
a graph is its \emph{support} $V_G$. Two graph $G,H \subseteq V
\times \Sigma \times V$ are isomorphic if there exists a bijection
$\rho$ from $V_H$ to $V_G$ such that for all $x,y \in V_H$, $x
\era{a}{H} y$ iff $\rho(x) \era{a}{G} \rho(y)$.

A sequence of edges $s_1 \erb{a_1} t_1, \ldots, s_k \erb{a_k} t_k$
with $\forall i \in [2,k],\ s_i = t_{i-1}$ is a \intro{path}. It is
written $s_1 \erb{u} t_k$, where $u = a_1 \ldots a_k$ is the
corresponding \intro{path label}.  A graph is \intro{deterministic} if
it contains no pair of edges with the same source and label. One can
relate a graph to a language by considering its path language, defined
as the set of all words labeling a path between two given sets of
vertices.

\begin{definition}
  The (path) language of a graph $G$ between two sets of vertices $I$
  and $F$ is the set
  \[
  L(G,I,F)\ =\ \{\ w\ |\ s \era{w}{G} t,\ s \in I,\ t \in F \}.
  \]
\end{definition}

\subsection{Linearly bounded Turing machines}
\label{ssec:lbm}

We recall the definitions of context-sensitive languages and linearly
bounded Turing machines. A context-sensitive language is a set of
words generated by a grammar whose production rules are of the form
$\alpha \rightarrow \beta$ with $|\beta| \geq |\alpha|$. Such grammars
are called \emph{context-sensitive} or sometimes \emph{growing}. A
more operational definition of context-sensitive languages is as the
class of languages accepted by \emph{linearly bounded Turing
  machines} (LBMs).

\begin{definition}
  \label{def:lbm}
  A linearly bounded Turing machine is a tuple $M = ( \Gamma, \Sigma,
  [, ],$ $Q,$ $q_0,$ $F,$ $\delta)$, where
  \begin{itemize}
  \item $\Gamma$ is a finite set of \intro{tape
      symbols},
  \item $\Sigma \subseteq \Gamma$ is the \intro{input alphabet} which
    does not contain the symbol $\varepsilon$,
  \item $[$ and $] \notin \Gamma$ are the \intro{left} and
    \intro{right end-marker},
  \item $Q$ is a finite set (disjoint from $\Gamma$) of \intro{control
      states},
  \item $q_0 \in Q$ is the unique \intro{initial state},
  \item $F \subseteq Q$ is a set of \intro{final states},
  \item $\delta$ is a finite set of \intro{transition rules} of one of
    the forms:
    \begin{xalignat*}{3}
      pA \erb{} & qB\pm & p[ \erb{} & q[+ & p] \erb{}& q]-
    \end{xalignat*}
    with $p,q \in Q$, $A,B \in \Gamma$ and $\pm \in \{ +,- \}$.
  \end{itemize}
\end{definition}

As usual in the syntax of Turing machines, symbols $+$ and $-$ in
transition rules respectively denote a move of the read head to the
right and to the left.  The set of configurations $C_M$ of $M$ is the
set of words $uqv$ such that $q \in Q$, $v \not= \eps$ and $uv \in
[\Gamma^*]$. The transition relation $\era{}{M}$ is a subset of $C_M
\times C_M$ defined as:
\begin{align*}
  \era{}{M} \ =\ & \{\ (upAv,uBqv)\ |\ pA \erb{} qB+ \in \delta\ \}
  \\
  \cup\ & \{\ (uCpAv, uqCBv)\ |\ pA \erb{} qB- \in \delta\}
\end{align*}
We will simply write $\erb{}$ when $M$ is understood. For any input
word $w \in \Sigma^*$, the unique initial configuration is $[q_0w]$
and a final configuration $c_f$ is a configuration containing a
terminal control state. A word $w$ is accepted by $M$ if $[q_0w]
\erb{} c_f$ where $c_f$ is a final configuration. Quite naturally, $M$
is \intro{deterministic} if, from any configuration, at most one rule
can be applied. Formally, for all configurations $c$, $c_1$ and $c_2$
such that $c \erb{} c_1$ and $c \erb{} c_2$, then $c_1=c_2$.

Other definitions of linearly bounded machines do not use the border
symbols to constrain the head inside a portion of the tape whose size
equals the size of the input. Instead, they externally require that
the tape size used at any point during any computation be at most $k$
times the size of the input, where $k$ is a fixed constant. It is a
well-known fact that $k$ can be considered equal to $1$ without loss
of generality, and that these definitions are equivalent to the one
given above. An interesting open problem raised by Kuroda
\cite{Kuroda64} concerns deterministic context-sensitive languages,
which are the languages accepted by deterministic LBMs. It is not
known whether they coincide with non-deterministic context-sensitive
languages, as is the case for recursively enumerable or rational
languages.

Note that, contrary to unbounded Turing machines, it is sufficient to
only consider linearly bounded machines which always terminate, also
called \intro{terminating} machines. This is expressed by the
following proposition:

\begin{proposition}
  \label{prop:termlbm}
  For every linearly bounded Turing machine, there exists a
  terminating linearly bounded Turing machine recognizing the same
  language.
\end{proposition}

\begin{proof} 
  It suffices to show that, for every linearly bounded machine $M$
  whose configuration graph contains a cycle, there is an equivalent
  terminating machine $M'$, i.e. a machine which always terminates and
  accepts the same language as $M$. The total number of distinct
  configurations of $M$ during a run on a word of size $n$ is bounded
  by $k^n$, where $k$ is a constant depending on the size of the
  control state set and work alphabet of $M$. It is easy to see that a
  word $w$ accepted by $M$ must be accepted by at least one run of
  size less than $k^{|w|}$. Indeed, if the smallest run accepting $w$
  was longer than this bound, it would necessarily contain two
  occurrences of the exact same configuration, i.e. a cycle. By
  removing this cycle, one would obtain a shorter accepting run. Let
  $M'$ be the machine which simulates $M$ while incrementing a
  $n$-digit counter in base $k$ encoded in the alphabet and stops the
  run if the counter overflows.  By construction, $M'$ accepts the
  same language as $M$ and is terminating. \qed
\end{proof}

Similar arguments are used in \cite{Kuroda64} to show that the set of
words on which a linearly bounded machine has an infinite run is
context-sensitive. Another property of context-sensitive languages is
that they are closed under union and complement.

\begin{theorem}[\cite{Immerman88,Szelepcsenyi88}]
  \label{thm:csb}
  The context-sensitive languages over a finite alphabet $\Gamma$ form
  an effective Boolean algebra.
\end{theorem}

The notion of space-bounded Turing machine can be extended to any
bound $f: \mathbb{N} \mapsto \mathbb{N}$ such that for all $n \geq 0$,
$f(n) \geq \log(n)$. The set of all languages accepted by a Turing
machine working in space $f(n)$ on an entry of size $n$ is written
$\mathrm{NSPACE}[f(n)]$. The following theorem states that the space
hierarchy is strict.

\begin{theorem}[\cite{Hopcroft79,Immerman88}]
  \label{thm:hierar}
  For all space-constructible $f$ and $g$ in $\mathbb{N} \mapsto
  \mathbb{N}$ such that $\lim_{n \rightarrow +\infty}
  \dfrac{f(n)}{g(n)}=0$, we have $\mathrm{NSPACE}[f(n)] \subsetneq
  \mathrm{NSPACE}[g(n)]$.
\end{theorem}

In particular, the class $\mathrm{NSPACE}[2^n]$ of languages
recognizable in exponential space strictly contains the class of
context-sensitive languages (or $\mathrm{NSPACE}[n]$).

\subsection{Rational graphs}

Consider the product monoid $\Sigma^* \times \Sigma^*$, whose elements
are pairs of words $(u,v)$ in $\Sigma^*$, and whose composition law is
defined by $(u_1,v_1) \cdot (u_2,v_2) = (u_1u_2, v_1v_2)$. A finite
transducer is an automaton over $\Sigma^* \times \Sigma^*$ with labels
in $(\Sigma \cup \{\eps\}) \times (\Sigma \cup \{\eps\})$. Transducers
accept the rational subsets of $\Sigma^* \times \Sigma^*$, which are
seen as binary relations on words and called rational transductions.
We do not distinguish a transducer from the relation it accepts and
write $(w,w') \in T$ if the pair $(w,w')$ is accepted by $T$.  Graphs
whose vertices are words and whose edge relations are defined by
transducers (one per letter in the label alphabet) are called rational
graphs.

\begin{definition}[\cite{Morvan00}]
  A rational graph labeled by $\Sigma$ with vertices in $\Gamma^*$ is
  given by a class of transducers $(T_a)_{a \in \Sigma}$ over
  $\Gamma$. For all $a \in \Sigma$, $(u,a,v) \in G$ if and only if
  $(u,v) \in T_a$.
\end{definition}

For $w \in \Sigma^+$ and $a \in \Sigma$, we recursively define $T_{wa}
= T_w \circ T_a$, where $\circ$ denotes the standard relational
composition, and we write $u \erb{w} v$ if and only if $(u,v) \in
T_w$. In general, there is no bound on the size difference between
input and output in a transducer (and hence between the lengths of two
adjacent vertices in a rational graph). Interesting sub-classes are
obtained by enforcing some form of synchronization. The most
well-known was defined by Elgot and Mezei \cite{Elgot65} (see also
\cite{FrougnyS93}) as follows. A transducer over $\Sigma$ with initial
state $q_0$ is (left-)synchronized if for every path $q_0
\erb{x_1/y_1} q_1 \ldots q_{n-1}$ $\erb{x_n/y_n} q_n$, there exists $k
\in [0,n]$ such that for all $i \in [1,k]$, $x_i$ and $y_i$ belong to
$\Sigma$ and either $x_{k+1} \ldots x_n = \eps$ and $y_{k+1} \ldots
y_n \in \Sigma^*$ or $y_{k+1} \ldots y_n = \eps$ and $x_{k+1} \ldots
x_n \in \Sigma^*$. A rational graph defined by synchronized
transducers will simply be called a synchronized (rational) graph.

Rational graphs form a class of infinite acceptors for
context-sensitive languages.  For a discussion on the expressive power
of the sub-classes of rational graph seen as language acceptors see
\cite{Carayol05}.

\begin{theorem}[\cite{Morvan01a}]
\label{thm:traces}
  The languages accepted by the rational graphs between a rational set
  of initial vertices and a rational set of final vertices are the
  context-sensitive languages.
\end{theorem}

This result can be slightly strengthened in the case of a single
initial and final vertex.

\begin{corollary}
  \label{cor:traces}
  For any rational graph $G$ labeled by $\Sigma$, pair $(i,f)$ of
  vertices of $G$, and symbol $\sharp \not\in \Sigma$, the language
  $L_G = \{ i\sharp w \sharp f \; | \; w \in L(G,\{i\},\{f\}) \}$ is
  context-sensitive.
\end{corollary}

\begin{proof}
  Let $G$ be a rational graph labeled by $\Sigma$ with vertices in
  $\Gamma^*$ and defined by a family of transducers $(T_a)_{a \in
    \Sigma}$.  Let $\bar{\Gamma}$ and $\tilde{\Gamma}$ be two finite
  alphabets disjoint from but in bijection with $\Gamma$. For any $x
  \in \Gamma$, we write $\bar{x}$ (resp. $\tilde{x}$) the
  corresponding symbol in $\bar{\Gamma}$ (resp.  $\tilde{\Gamma}$). We
  consider the rational graph $H$ labeled by $\Xi = \Sigma \cup
  \bar{\Gamma} \cup \tilde{\Gamma}$ defined by the family of
  transducers $(T_x)_{x \in \Xi}$ where for all $x \in \Gamma$,
  $T_{\bar{x}} = \{ (u,ux) \; | \; u \in \Gamma^* \}$ and
  $T_{\tilde{x}} = \{ (xu,u) \; | \; u \in \Gamma^* \}$.

  By Thm. \ref{thm:traces}, the language $L=L(H,\eps,\eps) \cap
  \bar{\Gamma}^* \Sigma^* \tilde{\Gamma}^*$ is context-sensitive. By
  construction, $L$ is equal to $\{ \bar{i}w\tilde{f} \; | \; i,f \in
  \Gamma^* \; \textrm{and} \; w \in L(G,\{i\},\{f\}) \}$.  It follows
  that $L_G$ is context-sensitive. \qed
\end{proof}

\section{Linearly Bounded Graphs}
\label{sect:lbg}

\subsection{LBM Transition Graphs}

Following \cite{Caucal03tm}, we define the notion of labeled linearly
bounded Turing machine (LLBM). This notion is very close to the notion
of off-line Turing machine \cite{Hopcroft79}. It is essentially
equivalent to an off-line Turing machine with a one-way input tape and
a two-way linearly bounded work tape. As in standard definitions of
LBMs, the transition rules can only move the head of the LLBM between
the two end markers $[$ and $]$. In addition, a silent step can
decrease the size of the configuration (without removing the markers)
and a $\Sigma$-transition can increase the size of the configuration
by one cell. This ensures that while reading a word of length $n$, the
labeled LBM uses at most $n$ cells.

\begin{definition}
  A labeled linearly bounded Turing machine is a tuple $M = ( \Gamma,
  \Sigma, [, ],Q, q_0,$ $F, \delta)$ as in Definition \ref{def:lbm},
  where all components but $\delta$ are defined similarly, and
  $\delta$ is a finite set of \emph{labeled} transition rules of one
  of the forms:
  \begin{xalignat*}{3}
    pA \erb{\eps} & qB\pm & p[ \erb{\eps} & q[+ & p] \erb{\eps}& q]-
    \\
    pB \erb{a} & qAB & p] \erb{a} & qA] & pA \erb{\eps} & q
  \end{xalignat*}
  with $p,q \in Q$, $A,B \in \Gamma$, $\pm \in \{ +,- \}$ and $a \in
  \Sigma$.
\end{definition}

Configurations are defined similarly to the previous case. However,
the transition relation is now labeled. For all $x \in \Sigma \cup
\{\eps\}$, the relation $\era{x}{M}$ is a subset of $C_M \times C_M$
defined as:
\begin{align*}
  \era{x}{M} \ =\ & \{\ (upAv,uBqv)\ |\ pA \erb{x} qB+ \in \delta\ \}
  \\
  \cup\ & \{\ (uCpAv, uqCBv)\ |\ pA \erb{x} qB- \in \delta\}
  \\
  \cup\ & \left\{
    \begin{aligned}
      \phantom{t}& \{\ (upAv,uqv) \ |\ pA \erb{x} q \in \delta\}
      \qquad & &\text{with $x=\eps$}
      \\
      & \{\ (upAv,uqBAv) \ |\ pA \erb{x} qBA \in \delta\} \qquad &
      &\text{with $x \in \Sigma$}.
    \end{aligned}
  \right.
\end{align*}
Transitions are composed by defining $\erb{wx}$ as $(\erb{w} \circ
\erb{x})$ for all $w \in \Sigma^*$. The definition of runs also
changes to reflect the fact that input words are no longer present on
the tape but are read as the run progresses. The unique initial
configuration is thus $[q_0]$. A word $w$ is accepted by $M$ if $[q_0]
\erb{w} c_f$ where $c_f$ is a final configuration. $M$ is
\intro{deterministic} if, from any configuration, either all possible
moves are labeled by \emph{distinct} letters of $\Sigma$, or there is
only one possible move labeled by $\eps$. Formally, for all
configurations $c$, $c_1$ and $c_2$ such that $c \erb{x} c_1$ and $c
\erb{y} c_2$, either $x$ and $y$ belong to $\Sigma$ and if $c_1 \not=
c_2$ then $x \not= y$, or $x=y=\eps$ and $c_1=c_2$.

Labeled LBMs are as expressive as classical LBMs.

\begin{proposition}
\label{prop:eqlbmllbm}
  A language is (deterministic) context-sensitive if and only if it is
  accepted by a (deterministic) labeled linearly bounded Turing
  machine.
\end{proposition}

\begin{proof}
  Let us first consider a context-sensitive language $L$ accepted by a
  terminating LBM $N$ (by Prop. \ref{prop:termlbm}). We sketch the
  construction of a LLBM $M=(\Gamma,\Sigma,[,],Q,q_0,F,\delta)$
  accepting $L$.
  % \footnote{We chose to simulate $N$ using a LLBM $M$ with a single
  %   working tape to avoid defining multi-tape LLBMs, which would not
  %   have been used anywhere else.}
  Let $q_A,q_R$ and $q_S$ be three states in $Q$. In a configuration
  of the form $[wq_X]$ for $w \in \Sigma^*$ and $q_X \in \{q_A,q_R\}$,
  $M$ reads an input letter $a \in \Sigma$ and goes to the
  configuration $[waq_s]$.  Then it simulates $N$ on $wa$ using only
  $\varepsilon$-transitions while remembering $wa$. Using the same
  work alphabet as $N$, this would require the use of $2|wa|$ cells.
  However using an increased work alphabet (or a second work tape),
  this can be done using only $|wa|$ cells. If $N$ accepts (resp.
  rejects) $wa$, $M$ restores $wa$ on its (primary) work tape and
  steps in configuration $[waq_A]$ (resp. $[waq_R]$). By taking
  $F=\{q_A\}$ and $q_0=q_A$ if $\varepsilon$ belongs to $L$ and
  $q_0=q_R$ otherwise, it is easy to see that the set of words
  accepted by $M$ from $[q_0]$ is precisely $L$. Note that
  non-deterministic behavior can only appear in $N$ while simulating
  $M$. Hence, if $M$ is deterministic then so is $N$.
  
  Conversely, let $M$ be a LLBM accepting a language $L \subseteq
  \Sigma^*$. We describe a LBM $N$ accepting $L$ with two tapes: the
  input tape and work tape. It is well known that this model is
  equivalent to LBMs with one tape as presented in Sec. \ref{ssec:lbm}
  (see for example \cite{Hopcroft79}). To each tape corresponds a set
  of control states: the set of work states $Q_\Gamma$ contains the
  set $Q$ of states of $M$ and the set of input states $Q_\Sigma$ is
  reduced to $\{q_i,q_A\}$ respectively the initial and accepting
  state.
  
  \noindent
  The machine $N$ working on word $w \in \Sigma^*$ starts with the
  configuration $([q_i w],[q_0])$. From a configuration
  $([w_1q_iw_2],[uqv])$ with $w_1,w_2 \in \Sigma^*$, $q \in Q$ and $u,
  v \in \Gamma^*$, $N$ can non-deterministically simulate any
  $\varepsilon$-transition of $N$ that can be applied to the
  configuration represented by the work tape. The deletion rules are
  simulated by shifting all tape content on the right of the read head
  by one cell to the left. Moreover $N$ can non-deterministically
  simulate any $a$-transition for $a \in \Sigma$ provided that the input
  head is on top of the symbol $a$ in which case it is moved one cell to the
  right.  The insertion rules are simulated by shifting the work tape
  content to the right of the read head by one cell on the right.
  
  \noindent
  The machine $M$ enters the accepting state $q_A$ if the input head
  is on top of the right border symbol $]$ and the work state is a
  final state of $N$. It follows from the construction of $M$ that $w$
  is accepted by $M$ if and only if it is accepted by $N$.  Moreover,
  if $N$ is deterministic then so is $M$. \qed
\end{proof}

\begin{remark}
  \label{rem:lbminit}
  For convenience, one may consider LBMs whose initial configuration
  is not of the form $[q_0]$ but is any fixed configuration $c_0$.
  This does not add any expressive power, as can be proved by a simple
  encoding of $c_0$ into the control state set of the machine, which
  will not be detailed here.
\end{remark}

\begin{remark}
  \label{rem:lbmrules}
  For simplicity, the above definition forces the insertion of a new
  tape cell each time a letter is read. More relaxed forms where a
  cell deletion or rewriting can occur during an input may be
  considered without any consequence for the results. Similarly, rules
  which do not move the read head can be allowed.
\end{remark}

Let $M = (\Gamma, \Sigma, [, ], Q, q_0, F, \delta)$ be a LLBM, we
define its configuration graph
\[
C_M\ =\ \big\{ (c,a,c') \mid c \era{a}{M} c' \quad \text{for } a \in
\Sigma \cup \{ \eps \}\, \big\}.
\]
The vertices of this graph are all configurations of $M$, and its
edges denote the transitions between them, including
$\eps$-transitions. One may wish to only consider the behavior of $M$
from an external point of view, i.e. only looking at the sequence of
inputs. This means one has to find a way to conceal $\eps$-transitions
without changing the accepted language or destroying the structure.
One speaks of the \emph{transition graph} of an acceptor, as opposed
to its configuration graph.

In \cite{Stirling00}, Stirling mentions a normal form for pushdown
automata which allows him to consider a structural notion of
transition graphs, without relying on the naming of vertices. We first
recall this notion of \emph{normalized} systems adapted to labeled
LBMs.  A labeled LBM is \emph{normalized} if its set of control states
can be partitioned in two subsets: one set of \emph{internal} states,
noted $Q_\eps$, which can perform $\eps$-rules and only $\eps$-rules,
and a set of \emph{external} states noted $Q_\Sigma$, which can only
perform $\Sigma$-rules. More formally:

\begin{definition}
  \label{def:norm}
  A labeled LBM $M = (\Gamma, \Sigma, [, ], Q, q_0, F, \delta)$ is
  \emph{normalized} if there are disjoint sets $Q_\Sigma$ and $Q_\eps$
  such that $Q = Q_\eps \cup Q_\Sigma$, $F \subseteq Q_\Sigma$, and
  \begin{align*}
    & pB \erb{a} qAB \in \delta \implies p \in Q_\Sigma,
    \\
    & pA \erb{\eps} qB\pm \in \delta \text{ or } pA \erb{\eps} q \in
    \delta \implies p \in Q_\eps,
    \\
    & p \in Q_\eps \implies \; \textrm{for all $A \in \Gamma$}, \;
    \exists pA \erb{\eps} qB\pm \in \delta.
  \end{align*}
\end{definition}

This definition implies in particular that a control state from which
there exists no transition must belong to $Q_\Sigma$. A configuration
is external if its control state is in $Q_\Sigma$, and internal
otherwise. This makes it possible to \emph{structurally} distinguish
between internal vertices, which have one or more outgoing
$\eps$-edges, and external ones which only have outgoing
$\Sigma$-edges or have no outgoing edges. Given any labeled LBM, it
is always possible to normalize it without changing the accepted
language.

\begin{proposition}
  Every labeled linearly bounded machine can be normalized without
  changing the accepted language.
\end{proposition}

\begin{proof}
  Let $M$ be a LLBM, we know from Prop.~\ref{prop:eqlbmllbm} that the
  language $L$ accepted by $M$ is context-sensitive. Consider the LLBM
  $M'$ accepting $L$ obtained in the construction of the proof of
  Prop.~\ref{prop:eqlbmllbm}. It is easy to see that $M'$ is
  normalized: $q_A$ and $q_R$ are the only states that can perform
  transitions labeled by $\Sigma$, and they cannot perform
  $\varepsilon$-transitions. All other states can only perform
  $\eps$-transitions. To verify the third condition in the definition
  of normalized LLBMs, it suffices to add instruction $pA \erb{\eps}
  pA$ for every $p$ and $A$ for which the condition is violated.
  Moreover, the unique terminal state is $q_A$, which is external.
%  Let $M$ be a LLBM, we build a normalized LLBM $M'$ equivalent to
%  $M$. Without loss of generality, we will assume $M$ has no
%  $\eps$-labeled loop containing an accepting configuration. This
%  property is very similar to real-timeness (Cf. Prop.
%  \ref{prop:termlbm}). First, $M'$ has at least all control states and
%  $\eps$-transitions of $M$, and the same initial state $q_0$. By
%  definition of LLBMs, $q_0$ already satisfies the definition of an
%  external state. For all transition rule $pB \erb{a} qAB$ of $M$ with
%  $p \neq q_0$, we add a new control state $p'$ and a rule $p'B
%  \erb{a} qAB$, as well as $\eps$-rules allowing to go from $p$ to
%  $p'$ without moving the read head. We then take as set of external
%  states for $M'$ the set of all such new states $p'$, which by
%  definition only allow $\Sigma$-transitions, as well as $q_0$ and the
%  states of $M$ which are source of no transition. This way, the
%  condition on control states is met. It remains to ensure that all
%  final states of $M'$ are external. To achieve this, whenever a final
%  state is encountered, we transmit this information throughout each
%  possible $\eps$-transition by marking the control state. The marking
%  is reset whenever an external state is met. By the previous
%  assumption on the absence of $\eps$-loops containing accepting
%  configurations, all computations reaching a final configuration will
%  eventually reach an external state of $M'$. It thus only remains to
%  declare all marked external states as accepting for $M'$ to be
%  normalized and equivalent to $M$.
  \qed
\end{proof}

\begin{remark}
  Normalization as previously described preserves determinism.
\end{remark}

From this point on, unless otherwise stated, we will only consider
normalized LLBMs. We can now define the transition graph of a LLBM as
the $\eps$-closure of its configuration graph, followed by a
restriction to its set of external configurations (which happens to be
a rational set).

\begin{definition} 
  Let $M = (\Gamma, \Sigma, [, ], Q, q_0, F, \delta)$ be a
  (normalized) LLBM, and $C_\Sigma$ be its set of external
  configurations. The transition graph of $M$ is
  \[
  G_M\ =\ \big\{ (c,a,c') \mid c,c' \in C_\Sigma,\ a \in \Sigma,\
  \land\ c \era{a\eps^*}{M} c'\, \big\}.
  \]
\end{definition}

We now define the class of linearly bounded graphs as the closure
under isomorphism of transition graphs of labeled LBMs.

\begin{example}
  \label{ex:lbm}
  Figure \ref{fig:exlbm} shows the transition graph of the normalized
  LLBM $M$ whose rules are:
  \begin{align*}
    q_0] & \erb{a} q_0a] & q_1a & \erb{b} q_1b+ & q_2b & \erb{a} q_3a- &
    q_3b & \erb{a} q_3a-
    \\
    q_0a & \erb{a} q_0aa & q_1] & \erb{\eps} q_2]- & & & q_3[ &
    \erb{\eps} q_1[+
    \\
    q_0a & \erb{b} q_1b+
  \end{align*}
  and whose unique accepting state is $q_2$. This machine accepts the
  language $\{ (a^nb^n)^+ \mid n \geq 1 \}$, which is also the
  language of paths of the graph between vertex $[q_0]$ and the set
  $[b^*q_2b]$. For the sake of clarity, only the part of the graph
  reachable from configuration $[q_0]$ is shown. We will see in Sect.
  \ref{sect:prop} that this sub-graph is still a linearly bounded
  graph.
\end{example}

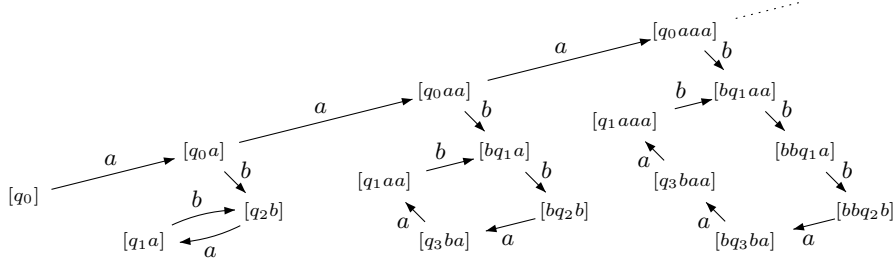
\begin{figure}
  \begin{center}
    \unitlength=2.3pt
    \begin{picture}(140,40)(10,-5)
      \gasset{Nframe=n,Nadjust=wh,Nadjustdist=2}
      
      \node(e)(10,2.5){\scriptsize$[q_0]$}
      \node(a1)(40,10){\scriptsize$[q_0a]$}
      \node(a2)(80,20){\scriptsize$[q_0aa]$}
      \node(a3)(120,30){\scriptsize$[q_0aaa]$}
      \node(a4)(140,35){}

      \node(b)(50,0){\scriptsize$[q_2b]$}
      \node(a)(30,-5){\scriptsize$[q_1a]$}

      \node(ab)(90,10){\scriptsize$[bq_1a]$}
      \node(bb)(100,0){\scriptsize$[bq_2b]$}
      \node(ba)(80,-5){\scriptsize$[q_3ba]$}
      \node(aa)(70,5){\scriptsize$[q_1aa]$}

      \node(aab)(130,20){\scriptsize$[bq_1aa]$}
      \node(abb)(140,10){\scriptsize$[bbq_1a]$}
      \node(bbb)(150,0){\scriptsize$[bbq_2b]$}
      \node(bba)(130,-5){\scriptsize$[bq_3ba]$}
      \node(baa)(120,5){\scriptsize$[q_3baa]$}
      \node(aaa)(110,15){\scriptsize$[q_1aaa]$}

      \drawedge(e,a1){$a$}
      \drawedge(a1,a2){$a$}
      \drawedge(a2,a3){$a$}

      \drawedge(a1,b){$b$}
      \drawedge[curvedepth=2](b,a){$a$}
      \drawedge[curvedepth=2](a,b){$b$}

      \drawedge(a2,ab){$b$}
      \drawedge(ab,bb){$b$}
      \drawedge(bb,ba){$a$}
      \drawedge(ba,aa){$a$}
      \drawedge(aa,ab){$b$}

      \drawedge(a3,aab){$b$}
      \drawedge(aab,abb){$b$}
      \drawedge(abb,bbb){$b$}
      \drawedge(bbb,bba){$a$}
      \drawedge(bba,baa){$a$}
      \drawedge(baa,aaa){$a$}
      \drawedge(aaa,aab){$b$}

      \gasset{Nadjustdist=6, AHnb=0,dash={0.3 1}0}
      \drawedge(a3,a4){}
    \end{picture}  
  \end{center}
  \caption{The transition graph of a labeled LBM accepting $\{
    (a^nb^n)^+ \mid n \geq 1 \}$.}
  \label{fig:exlbm}
\end{figure}

\subsection{Alternative definitions}

This section provides two alternative definitions of linearly bounded
graphs. In \cite{Caucal02}, it is shown that all previously mentioned
classes of graphs can be expressed in a uniform way in terms of
Cayley-type graphs of certain classes of rewriting systems.  We show
that it is also the case for linearly bounded graphs, which are the
Cayley-type graphs of length-decreasing rewriting systems. The second
alternative definition we present changes the perspective and directly
defines the edges of linearly bounded graphs using incremental
context-sensitive transductions.  This variety of definitions will
allow us to prove in a simpler way some of the properties of linearly
bounded graphs.

\subsubsection{Cayley-type graphs of decreasing rewriting systems}

We first give the relevant definitions about rewriting systems and
Cayley-type graphs. A \intro{word rewriting system} $R$ over alphabet
$\Gamma$ is a subset of $\Gamma^* \times \Gamma^*$. Each element
$(l,r) \in R$ is called a \intro{rewriting rule} and noted $l
\rightarrow r$. The words $l$ and $r$ are respectively called the
left- and right-hand side of the rule.  The \intro{rewriting relation}
of $R$ is the binary relation
\[
\{ (ulv, urv) \mid u,v \in \Gamma^*,\ l \rightarrow r \in R \}
\]
which we also denote by $R$, consisting of all pairs of words
$(w_1,w_2)$ such that $w_2$ can be obtained by replacing
(\emph{rewriting}) an instance of a left-hand side $l$ in $w_1$ with
the corresponding right-hand side $r$. The reflexive and transitive
closure $R^*$ of this relation is called the \intro{derivation} of
$R$. Whenever for some words $u$ and $v$ we have $u R^* v$, we say $R$
rewrites $u$ into $v$. A word which contains no left-hand side is
called a \emph{normal form}. The set of all normal forms of $R$ is
written $\mathrm{NF}(R)$.

One can associate a unique infinite graph to any rewriting system by
considering its \emph{Cayley-type graph} defined as follows:

\begin{definition}
  The $\Sigma$-labeled Cayley-type graph of a rewriting system $R$
  over $\Gamma$, with $\Sigma \subseteq \Gamma$, is the infinite graph
  \[
  G_R\ =\ \{ (u,a,v) \mid a \in \Sigma,\ u,v \in \mathrm{NF}(R),\ ua
  R^* v \}.
  \]
\end{definition}

The class of rewriting systems we consider is that of \emph{finite
  length decreasing word rewriting systems}, i.e. rewriting systems
with a finite set of rules of the form $l \rightarrow r$ with $|l|
\geq |r|$, which can only preserve or decrease the length of the word
to which they are applied. The reason for this choice is that the
derivation relation of such a system coincides with arbitrary
compositions of $\eps$-rules of a given labeled LBM.

\begin{example}
  \label{ex:cayley}
  Figure \ref{fig:excayley} shows the Cayley-type graph of a simple
  decreasing rewriting system.
\end{example}

\begin{figure}
  \begin{center}
    \unitlength=0.9mm
    \begin{picture}(120,55)(0,0)
      \gasset{Nframe=n,Nadjust=wh,Nadjustdist=2}

      \put(60,52.5){
        \node(   e)(  0,  0){\small $\eps$}
        \node(   0)(-32,-15){\small $0$}
        \node(   1)( 32,-15){\small $1$}
        \node(  00)(-48,-30){\small $00$}
        \node(  01)(-16,-30){\small $01$}
        \node(  10)( 16,-30){\small $10$}
        \node(  11)( 48,-30){\small $11$}
        \node(000)(-56,-45){\small $000$}
        \node(001)(-40,-45){\small $001$}
        \node(010)(-24,-45){\small $010$}
        \node(011)( -8,-45){\small $011$}
        \node(100)(  8,-45){\small $100$}
        \node(101)( 24,-45){\small $101$}
        \node(110)( 40,-45){\small $110$}
        \node(111)( 56,-45){\small $111$}
        \node(0000)(-60,-52.5){}
        \node(0001)(-52,-52.5){}
        \node(0010)(-44,-52.5){}
        \node(0011)(-36,-52.5){}
        \node(0100)(-28,-52.5){}
        \node(0101)(-20,-52.5){}
        \node(0110)(-12,-52.5){}
        \node(0111)( -4,-52.5){}
        \node(1000)(  4,-52.5){}
        \node(1001)( 12,-52.5){}
        \node(1010)( 20,-52.5){}
        \node(1011)( 28,-52.5){}
        \node(1100)( 36,-52.5){}
        \node(1101)( 44,-52.5){}
        \node(1110)( 52,-52.5){}
        \node(1111)( 60,-52.5){}

        \gasset{ELside=r}
        \drawedge(  e,  0){$a$}
        \drawedge(  0, 00){$a$}
        \drawedge(  1, 10){$a$}
        \drawedge( 00,000){$a$}
        \drawedge( 01,010){$a$}
        \drawedge( 10,100){$a$}
        \drawedge( 11,110){$a$}

        \gasset{ELside=l}
        \drawedge(  0,  1){$b$}
        \drawedge( 00, 01){$b$}
        \drawedge( 01, 10){$b$}
        \drawedge( 10, 11){$b$}
        \drawedge(000,001){$b$}
        \drawedge(001,010){$b$}
        \drawedge(010,011){$b$}
        \drawedge(011,100){$b$}
        \drawedge(100,101){$b$}
        \drawedge(101,110){$b$}
        \drawedge(110,111){$b$}

        \gasset{ELside=r}
        \drawedge(  1,  e){$c$}
        \drawedge( 01,  0){$c$}
        \drawedge( 11,  1){$c$}
        \drawedge(001, 00){$c$}
        \drawedge(011, 01){$c$}
        \drawedge(101, 10){$c$}
        \drawedge(111, 11){$c$}

        \gasset{Nadjustdist=6, AHnb=0,dash={0.3 1}0}
        \drawedge(000,0000){}
        \drawedge(001,0010){}
        \drawedge(010,0100){}
        \drawedge(011,0110){}
        \drawedge(100,1000){}
        \drawedge(101,1010){}
        \drawedge(110,1100){}
        \drawedge(111,1110){}
        \drawedge(000,0001){}
        \drawedge(001,0011){}
        \drawedge(010,0101){}
        \drawedge(011,0111){}
        \drawedge(100,1001){}
        \drawedge(101,1011){}
        \drawedge(110,1101){}
        \drawedge(111,1111){}
      }
    \end{picture}
  \end{center}
  \caption{Cayley-type graph of the rewriting system $R = \{a
    \rightarrow 0,$ $b \rightarrow b,$ $0b \rightarrow 1,$ $1b
    \rightarrow b0,$ $c \rightarrow c,$ $1c \rightarrow \eps\}$, with
    $\Sigma = \{a,b,c\}$ and $\Gamma = \{a,b,c,0,1\}$. Rules $b
    \rightarrow b$ and $c \rightarrow c$ ensure that no normal form of
    $R$ contains $b$ or $c$. As an example, edge $01 \erb{b} 10$ is
    justified by the derivation $01b \rightarrow 0b0 \rightarrow 10$.}
  \label{fig:excayley}
\end{figure}
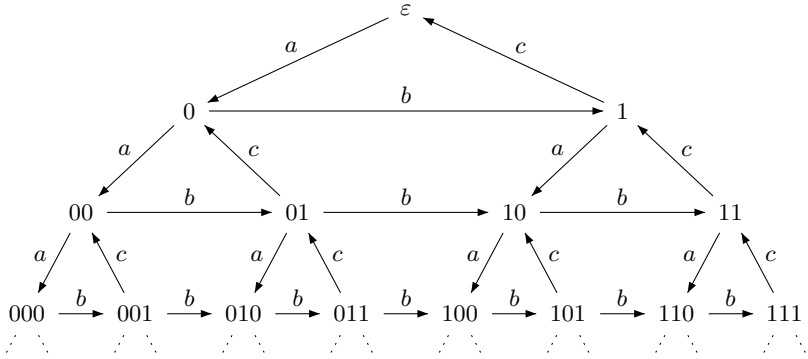

\subsubsection{Incremental context-sensitive transduction graphs}

The notion of \intro{computation graph} was first systematically used
in \cite{Caucal02} (where the terminology used is \intro{relation
  graph}). It corresponds to the graphs defined by the
\emph{transductions} (i.e.  binary relations on words) associated to a
class of finite machines.  These works prove that for pushdown
automata and Turing machines, the classes of transition and
computation graphs coincide. We show that it is also the possible to
give a definition of linearly bounded graphs as the computation graphs
of a certain class of LBMs, or equivalently as the graphs defined by a
certain class of context-sensitive transductions.

A relation $R$ is recognized by a LBM $M$ if the language $\{ u\#v \;
| \; (u,v) \in R\}$ where $\#$ is a fresh symbol is accepted by $M$.
However, this type of transductions generates more than linearly
bounded graphs. Even if we only consider linear relations (i.e
relations R such that there exists $c$ and $k \in \mathbb{N}$ such
that $(u,v) \in R$ implies $|v| \leq c \cdot |u| +k$), we obtain
graphs accepting the languages recognizable in exponential space
($\mathrm{NSPACE}[2^n]$) which strictly contain the context-sensitive
languages (cf Thm.~\ref{thm:hierar}). We need to consider relations
for which the length difference between a word and its image is
bounded by a certain constant. Such relations can be associated to
LBMs as follows:

\begin{definition}
  A \intro{$k$-incremental context-sensitive transduction} $T$ over
  $\Gamma$ is defined by a LBM recognizing a language $L \subseteq
  \{u\#v \; | \; u,v \in \Gamma^* \; \text{and} \; |v| \leq |u|+k\}$
  where $\#$ does not belong to $\Gamma$. Relation $T$ is defined as
  $\left\{ (u,v) \; | \; u\#v \in L \right\}$.
\end{definition}

The synchronized relations of finite image (i.e for $u$ there are
finitely many $v$ such that $(u,v) \in R$) provide a first example of
$k$-incremental context-sensitive transductions.

 \begin{proposition}
\label{prop:synchronizedincremental}
   For any synchronized relation $R$ of finite image, there
   exists a constant $k \in \mathbb{N}$ such that $R$ is a
   $k$-incremental context-sensitive transduction.
 \end{proposition}

 \begin{proof}
   Let $R \subseteq \Gamma^* \times \Gamma^*$ be a synchronized
   relation of finite image. It follows from the definition of
   synchronized relations that $R$ is equal to a finite union:
   \[
   \left(\bigcup_{i \in I} S_i \cdot (A_i \times \varepsilon) \right)
   \; \cup \; \left(\bigcup_{j \in J} S_j \cdot (\varepsilon \times
     B_j) \right)
   \]
   where for all $k \in I \cup J$, $S_k$ is a length-preserving
   relation\footnote{A synchronized relation $R$ is length-preserving
     if $\forall (x,y) \in R, |x| = |y|$} and for all $i \in I$ and $j
   \in J$, $A_i$ and $B_j$ are rational subsets of $\Gamma^*$.

   \noindent
   As $R$ has a finite image, for all $j \in J$ the set $B_j$ is
   necessarily finite. Let $k$ be the maximal length of a word in
   $\bigcup_{j \in J} B_j$, it is easy to check that for all pairs of
   words $(u,v) \in R$, $|v| \leq |u|+k$.  Moreover the language $\{
   u\#v \; | \; (u,v) \in R \}$ is context-sensitive.  Hence $R$ is a
   $k$-incremental context-sensitive transduction. \qed
\end{proof}

 The following proposition states that incremental context-sensitive
 transductions of a given level form a Boolean algebra.

 \begin{proposition}
   \label{prop:bool}
   For all $k$-incremental context-sensitive transductions $T$ and
   $T'$ over $\Gamma^*$, $T \cup T',$ $T \cap T'$ and $\overline{T} =
   E_k - T$ (where $E_k$ is $\{ (u,v) \; | \; 0 \leq |v| \leq |u|+k \}$)
   are incremental context-sensitive transductions.
 \end{proposition}

 \begin{proof}  
   The closure under union follows from that of context-sensitive
   languages. The proof of closure under complement is a
   straightforward consequence of the closure under complement of
   context-sensitive languages (cf Thm.~\ref{thm:csb}). Let $T \subset
   \Gamma^* \times \Gamma^*$ be a $k$-incremental context-sensitive
   transduction. By definition, the set $L=\{ u\#v \;|\; (u,v) \in T
   \}$ is context-sensitive.  It is straightforward to check that the
   set $L'=\{ u\#v \; | \; (u,v) \in E_k - T \}$ is equal to:
   \[
   \overline{L} \cap \{ u \# v \; | \; |u| \leq |v|+k \}.
   \]
   As the context-sensitive languages are closed under complement and
   intersection, $L'$ is context-sensitive. Hence, $\bar{T}$ is a
   $k$-incremental context-sensitive transduction. \qed
\end{proof}

The canonical graph associated to a finite set of transductions is
called a \emph{transduction graph}. Relating graphs to a class of
binary relations on words was already used to define rational graphs
and their sub-classes.

\begin{definition}
  The $\Sigma$-labeled transduction graph of a finite set of
  incremental context-sensitive transductions $(T_a)_{a \in \Sigma}$
  is
  \[
  G_T\ =\ \{ (u,a,v) \mid a \in \Sigma \; \text{and} \; (u,v) \;
  \text{is recognized by $T_a$} \}.
  \]
\end{definition}

\begin{example}
  The linearly bounded graph of Ex. \ref{ex:lbm} is isomorphic to the
  transduction graph of the following set of incremental
  context-sensitive transductions:
  \begin{align*}
    T_a = \{(\#a^n,\#a^{n+1}) \mid n \geq 0\} & \cup
    \{(b^ma^n,b^{m-1}a^{n+1}) \mid m \geq 1, n \geq 0\},
    \\
    T_b = \{(\#a^n,a^{n-1}b) \mid n \geq 1\} & \cup
    \{(a^mb^n,a^{m-1}b^{n+1}) \mid m \geq 1, n \geq 0\}.
  \end{align*}
  The symbol $\#$ is needed to distinguish a vertex directly reachable
  through a sequence of $a$'s from a vertex reachable through a
  sequence of the form $(a^nb^n)^*$.

  As for the Cayley-type graph of Ex. \ref{ex:cayley}, it can be seen
  as the transduction graph of the set of incremental
  context-sensitive transductions $\{T_a, T_b, T_c\}$, where $T_a$
  adds a 0 and $T_c$ removes a 1 to the right of a binary number, and
  $T_b$ implements binary increment.
\end{example}

Length-preserving context-sensitive transductions have already been
extensively studied in \cite{Latteux98}. In the rest of this
presentation, unless otherwise stated, we will only consider
$1$-incremental transductions without loss of generality regarding the
obtained class of graphs: indeed, any $k$-incremental transduction
graph is isomorphic to a $1$-incremental one over an increased
alphabet.

\subsubsection{Equivalence of all definitions}

We now prove that both classes of Cayley-type graphs of decreasing
rewriting systems, and incremental context-sensitive transduction
graphs define precisely the class of linearly bounded graphs, up to
isomorphism (i.e. up to vertex renaming).

\begin{theorem}
\label{thm:equiv}
  For any graph $G$, the following statements are equivalent:
  \begin{enumerate}
  \item $G$ is isomorphic to the transitions graph of a labeled LBM,
  \item $G$ is isomorphic to the Cayley-type graph of a finite
    length-decreasing system,
  \item $G$ is isomorphic to a context-sensitive transduction graph.
  \end{enumerate}
\end{theorem}

\begin{proof}
  $1 \implies 2$: \quad Let $M = (\Gamma, \Sigma, [, ], Q, q_0, F,
  \delta)$ be a normalized labeled linearly bounded Turing machine,
  with $\Sigma \cap \Gamma = \emptyset $. As $M$ is normalized, its
  control states can be partitioned into $Q_\Sigma$ and $Q_\eps$ (see
  Definition \ref{def:norm}). Let $\Gamma' = \Gamma \cup \{[,]\}$.  We
  build a finite length-decreasing rewriting system $R$ whose
  Cayley-type graph is the transition graph of $M$.  Let the alphabet
  of $R$ be $\Delta = \Sigma \cup \Gamma' \cup (\Gamma' \times Q) \cup
  S$, where $S = \{ v_a, v_a', s_a \mid a \in \Sigma \}$ is a new set
  of symbols disjoint from $\Gamma$ and $Q$. Elements $(x,q)$ of
  $\Gamma' \times Q$ will be noted $x_q$. For convenience, for any set
  $X$, $X_\bullet$ will denote $X \cup (X \times Q_\Sigma)$, and
  $x_\bullet$ any symbol in $\{x\}_\bullet$.

  There are several important points which the rules of $R$ must
  ensure:
  \begin{enumerate}
  \item Only words of the form $(\eps \, \cup \, [_\bullet) \,
    \Gamma_\bullet^* \, (\eps \, \cup \, ]_\bullet)$, i.e.  words in
    which only external control states may occur and in which no
    symbol occurs to the left of a left bracket or to the right of a
    right bracket, should be normal forms (please be aware that
    $\Gamma_\bullet$ denotes the set $\Gamma \cup (\Gamma \times
    Q_\Sigma)$, and not $\Gamma \cup (\Gamma \times Q)$):
    \[
    x[_\bullet \rightarrow [_\bullet, \qquad ]_\bullet y \rightarrow
    ]_\bullet, \qquad s \rightarrow s
    \]
    for all $x \in \Delta,\ y \in \Delta \setminus \Sigma,\ s \in
    \Sigma \cup S \cup (\Gamma' \times Q_\eps)$.\medskip
  \item When a letter $a$ in $\Sigma$ is added to the right of an
    irreducible word $u$, one should ensure that $u$ actually
    represents a legal configuration\footnote{Note that this encoding
      of configuration differs from the one used previously, where
      control states were standalone symbols.}, i.e. that $u$ is of
    one of the forms $[_q \Gamma^*]$, $[\Gamma^* A_q \Gamma^*]$ or
    $[\Gamma^*]_q$ for $A \in \Gamma$, $q \in Q$:
    \begin{gather*}
      ]a \rightarrow v_a], \qquad ]_q a \rightarrow {v_a}' ]_q, \qquad
      A v_a \rightarrow v_a A, \qquad A_q v_a \rightarrow {v_a}' A_q,
      \\
      A{v_a}' \rightarrow {v_a}'A, \qquad {[_q} {v_a} \rightarrow [
      s_a, \qquad {[} {v_a}' \rightarrow [ s_a.
    \end{gather*}
  \item Finally, once it has been made sure that the word represents a
    legitimate LBM configuration, one should simulate an insertion
    operation followed by any number of $\eps$-transitions of the LBM:
    \[
    s_a A \rightarrow A s_a, \qquad s_a B_p \rightarrow C_q B
    \]
    for all $a \in \Sigma,\ A,B,C \in \Gamma$ and $pB \erb{a} qCB \in
    \delta$, and
    \[
    A_pC \rightarrow BC_q, \qquad CA_p \rightarrow C_qB \qquad A_pC
    \rightarrow C_q
    \]
    for all $pA \erb{\eps} qB+,\ pA \erb{\eps} qB-,\ pA \erb{\eps} q
    \in \delta$ (respectively).
  \end{enumerate}
  There is an edge $u \erb{a} v$ in the Cayley-type graph $G_R$ of $R$
  if and only if $u$ and $v$ are words representing valid
  configurations of $M$ from which no $\eps$-transition can be
  performed, i.e.  observable configurations, and there exists a
  sequence of transitions labeled by $a \varepsilon^*$ of $M$ by
  which $u$ reaches $v$.  There is a bijection between the edges and
  vertices of $G_R$ and the transition graph of $M$, hence these two
  graphs are isomorphic.

  $2 \implies 3$: \quad Let $R$ be a finite length-decreasing
  rewriting system, $G_R$ its Cayley-type graph. For each letter $a$,
  we will show that the relation
  \[
  T_a\ =\ \{ (ua,v) \mid u \era{a}{G_R} v \} = \{ (ua,v) \mid u,v \in
  \mathrm{NF}(R) \land ua R^* v\}
  \]
  is an incremental context-sensitive transduction by building a LBM
  $M_a = (\Gamma, \Sigma, [, ], Q, q_0, F, \delta)$ recognizing $T_a$.

  For every pair $(ua,v)$, $M_a$ starts in configuration $ua \# v$, and
  first has to check that $u$ is a normal form of $R$ by verifying
  that it contains no left-hand side of any rule in $R$. Second, $M_a$
  simulates the derivation of $R$ on $ua$, applying one rewriting rule
  at a time until a normal form is reached. Due to non-determinism,
  there might be unsuccessful runs, but the pair is accepted if and
  only if one run reaches configuration $v \# v$, meaning that $R$ can
  normalize $ua$ into $v$. Hence, a pair $(ua,v)$ is in $T_a$ if and
  only if $(u,a,v) \in G_R$, meaning that the transduction graph of
  $(T_a)_{a \in \Sigma}$ is isomorphic to $G_R$.

  $3 \implies 1$: \quad Let $T = (T_a)_{a \in \Sigma}$ be a finite set
  of incremental context-sensitive transductions defining a graph
  $G_T$, each $T_a$ being recognized by an LBM $M_a$. We informally
  describe a normalized labeled LBM $M$ whose transition graph $G_M$
  is isomorphic to $G_T$.

  Let $q$ be the unique external control state of $M$, $M$ should have
  a run labeled by $a \eps^*$ between configurations $q u$ and $q v$
  whenever $(u, v) \in T_a$, or equivalently whenever the word $u\#v$
  is accepted by $M_a$. This is done as follows. First, starting from
  configuration $q u$, $M$ should perform an $a$-labeled transition,
  increasing its available tape space by 1, and step into an internal
  control state. It should then guess a word $v$, and write $u\#v$ on
  its tape. Since $T_a$ is incremental, this can be done using no more
  than $|u|+1$ tape cells by writing two symbols in each cell. Then,
  $M$ simulates the LBM $M_a$ on input word $u\#v$, while keeping an
  intact copy of $v$ on the tape (this can again be done by a simple
  alphabet encoding). If the simulated run of $M_a$ succeeds, $M$
  steps into external configuration $q v$ by restoring the saved copy
  of $v$ on the tape, otherwise it loops in a non-accepting internal
  state. By this construction, there is an edge $(qu,a,qv)$ in $G_M$
  if and only if there is an edge $(u,a,v)$ in $G_T$, hence both
  graphs are isomorphic. \qed
\end{proof}

This shows that the three types of graphs presented in this section
actually all define the same class, namely that of linearly bounded
graphs. This variety of definitions will allow us to prove in a
simpler way some of the properties of linearly bounded graphs.

\section{Structural properties}
\label{sect:prop}

Now that the class of linearly bounded graphs has been defined using
three different formalisms, we can easily deduce some of their
structural properties. In particular, we look at the languages
accepted by linearly bounded graphs and some of their closure
properties. We also give some insight about the relation between
linearly bounded graphs and deterministic context-sensitive languages,
and conclude with a few logical properties. But first, we compare our
notions to related work.

\subsection{Comparison with existing work}
\label{ssec:previouswork}

We now give a precise comparison of linearly bounded graphs with the
restriction of Turing graphs (\cite{Caucal03tm}) to the linearly
bounded case, and with the configuration graphs considered in
\cite{Knapik99}.

\subsubsection{Configuration graphs}

In \cite{Knapik99}, the configuration graphs of a class of offline
linearly bounded machines very similar to our labeled LBM are
considered. However, they include in their definition a restriction to
the set of configurations reachable from the initial configuration.
Therefore, their class of configuration graphs is incomparable to
ours.  As we will see in Prop.~\ref{prop:restriction}, the linearly
bounded graphs are closed under restriction to the set of vertices
reachable from a given vertex. Hence, it is not necessary to impose
this restriction directly in the definition of the configuration
graph.

Apart from this restriction, our class of configuration graphs
coincides with the configuration graphs of \cite{Knapik99} and with
those of \cite{Caucal03tm} in the case of linearly bounded Turing
machines.

\subsubsection{Transition graphs}

Knapik and Payet do not consider transition graphs: instead of
characterizing the $\eps$-closure of configuration graphs, they prove
a closure property of this class up to weak bisimulation
\cite{Milner89}, which is not a structural characterization.
Nevertheless, it is straightforward to prove that their results can be
extended to the class of transition graphs considered up to
isomorphism, as mentioned in Sect. \ref{sect:prop}.

In \cite{Caucal03tm}, Caucal defines the transition graphs of Turing
machines from their configuration graphs using the very general notion
of $\eps$-closure: wherever there is a path labeled by $\eps^* a
\eps^*$ in the configuration graph, there is an edge labeled by $a$ in
its $\eps$-closure, for every letter $a$. Furthermore, the definition
allows the restriction to an arbitrary rational set of vertices
(potentially the whole set of vertices). Figure \ref{fig:eps}
illustrates on a small configuration graph the difference between
these two approaches.

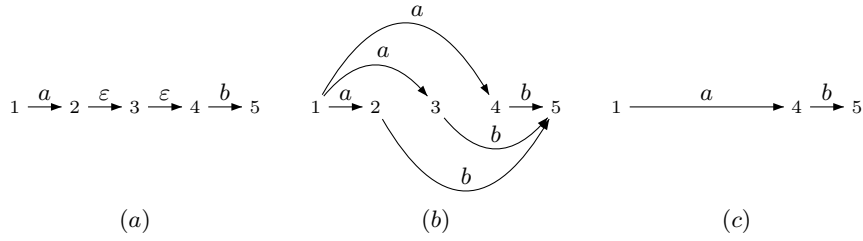
\begin{figure}
  \begin{center}
    \unitlength=0.8mm
  \begin{picture}(140,37)
    \gasset{Nframe=n,Nadjust=wh,Nadjustdist=1.3}

    \put(0,20){
      \node(e)(0,0){\scriptsize 1}
      \node(a)(10,0){\scriptsize 2}
      \node(ae)(20,0){\scriptsize 3}
      \node(aee)(30,0){\scriptsize 4}
      \node(aeeb)(40,0){\scriptsize 5}

      \node(legende)(20,-19){$(a)$}

      \drawedge(e,a){$a$}
      \drawedge(a,ae){$\eps$}
      \drawedge(ae,aee){$\eps$}
      \drawedge(aee,aeeb){$b$}
    }

    \put(50,20){
      \node(e)(0,0){\scriptsize 1}
      \node(a)(10,0){\scriptsize 2}
      \node(ae)(20,0){\scriptsize 3}
      \node(aee)(30,0){\scriptsize 4}
      \node(aeeb)(40,0){\scriptsize 5}

      \node(legende)(20,-19){$(b)$}

      \drawedge(e,a){$a$}
      \drawedge[ELpos=55,curvedepth=7](e,ae){$a$}
      \drawedge[ELpos=55,curvedepth=14](e,aee){$a$}
      \drawedge(aee,aeeb){$b$}
      \drawedge[curvedepth=-7](ae,aeeb){$b$}
      \drawedge[curvedepth=-14](a,aeeb){$b$}
    }

    \put(100,20){
      \node(e)(0,0){\scriptsize 1}
      \node(aee)(30,0){\scriptsize 4}
      \node(aeeb)(40,0){\scriptsize 5}

      \node(legende)(20,-19){$(c)$}

      \drawedge(e,aee){$a$}
      \drawedge(aee,aeeb){$b$}
    }
  \end{picture}
  \end{center}
  \caption{(a) Normalized LLBM configuration graph. (b) Transition
    graph of (a) as defined in \cite{Caucal03tm} (with no restriction
    on vertices). (c) Linearly bounded graph associated to (a).}
  \label{fig:eps}
\end{figure}

Our approach has two main advantages. First, the $\eps$-closure
operation as defined by Caucal may give rise to spurious
non-determinism in the obtained graphs. This additional complexity is
artificial and is eliminated using our definition. Furthermore, our
notion is purely structural, and does not rely on the naming of
vertices. But interestingly, we can show that both classes still
coincide up to isomorphism: for every labeled linearly bounded
machine $M$ and rational set $R$, one can define a machine $M'$ such
that the transition graph of $M$ following Caucal's definition with
respect to $R$ is isomorphic to the linearly bounded graph
corresponding to $M'$ in our framework.

\begin{proposition}
  The class of linearly bounded graphs coincides, up to isomorphism,
  with the restriction of Turing graphs to the (uniform) linearly
  bounded case.
\end{proposition}

\begin{proof}
  Let $M$ be a normalized LLBM, $C$ its configuration graph, $R$ its
  (rational) set of external configurations and $G$ its transition
  graph defined as
  \[
  G = \{ c \erb{a} c' \mid c,c' \in R \land c \erb{a \eps^*} c' \in C
  \}.
  \]
  Consider the graph $G'$ defined as the restriction to $R$ of the
  $\eps$-closure of $C$.
  \[
  G' = \{ c \erb{a} c' \mid c,c' \in R \land c \erb{\eps^* a \eps^*}
  c' \in C \}.
  \]
  The graphs $G$ and $G'$ are equal, since by definition of external
  vertices, they have no outgoing $\eps$-transitions.
  
  Conversely, let $M$ be a (not necessarily normalized) LLBM, $C$ its
  configuration graph, and $R$ a rational set of configurations of
  $M$. The binary relation $\erb{\eps^* a \eps^*}$ in $C \times C$ is
  a 1-incremental context-sensitive transduction. Hence the graph $G$
  defined as the restriction to $R$ (see Prop. \ref{prop:restriction})
  of the $\eps$-closure of $C$ is a linearly bounded graph. \qed
\end{proof}

\subsection{Languages}

It is quite obvious that the language of the transition graph of a
LLBM $M$ between the vertex representing its initial configuration and
the set of vertices representing its final configurations is the
language of $M$. In fact, the choice of initial and final vertices has
no importance in terms of the class of languages one obtains.

\begin{proposition}
  \label{prop:trace}
  The languages of linearly bounded graphs between an initial vertex
  $i$ and a finite set $F$ of final vertices are the context-sensitive
  languages.
\end{proposition}

\begin{proof}
  Let $G$ be a linearly bounded graph labeled by $\Sigma$ defined by a
  family $(T_a)_{a \in \Sigma}$ of $1$-incremental context-sensitive
  transductions. We will prove that $L(G,i,F)$ is context-sensitive
  even if $F$ is a context-sensitive set.
  
  Let $\# \not\in \Sigma$ be a new symbol, we consider the graph
  $\overline{G}$ obtained from $G$ by adding a loop labeled by $\#$
  on each vertex in $F$. Obviously, $\overline{G}$ is a linearly
  bounded graph as $\erb{\#}=\{(f,f) \;| \; f \in F \}$ is a
  $0$-incremental context-sensitive transduction. By Thm.
  \ref{thm:equiv}, $G$ is the transition graph of a LLBM $M$. Let
  $c_0$ be the configuration of $M$ corresponding to the vertex $i$ in
  $G$.  Consider the LLBM $M'$ whose initial configuration is $i$ (see
  Rem.  \ref{rem:lbminit}) and whose set of final states are the
  states that appear in the left-hand side of a $\#$-rule (without
  loss of generality, we can assume that the $\#$-rules do not depend
  on the current value of the cell). It is easy check that $M'$
  accepts $L(G,i,F)$ and by Prop. \ref{prop:eqlbmllbm}, $L(G,i,F)$ is
  context-sensitive.

  For the converse implication, let $L$ be a context-sensitive
  language and $M$ the normalized LLBM constructed in Prop.
  \ref{prop:eqlbmllbm} accepting $L$. The transition graph of $M$
  traces $L$ from the initial configuration to an infinite set of
  configurations (i.e the configurations with state $q_A$). Therefore,
  we need to adapt the construction of Prop. \ref{prop:eqlbmllbm}. We
  add a new state $q_f$ to the machine $M$ and whenever $M$ makes an
  $\eps$-transition to enter state $q_A$, we add the ability to enter
  configuration $[q_f]$ by a sequence of $\eps$-transitions.  As
  $[q_f]$ has no out-going edges, it is an external configuration and
  it is easy to see that $L=L(G,c_0,[q_f])$. \qed
\end{proof}

\begin{remark}
  When a linearly bounded graph is explicitly seen as the transition
  graph of a LLBM, as a Cayley-type graph or as a transduction graph,
  i.e. when the naming of its vertices is fixed, considering
  context-sensitive sets of final vertices does not increase the
  accepted class of languages.
\end{remark}

\subsection{Closure properties}

Linearly bounded graphs enjoy several good properties, which will be
especially important when comparing them to other classes of graphs
related to LBMs or context-sensitive languages (see Section
\ref{sect:comp}).

\begin{proposition}
\label{prop:restriction}
  The class of linearly bounded graphs is closed under restriction to
  reachable vertices from any vertex and under restriction to a
  context-sensitive set of vertices.
\end{proposition}

\begin{proof}
  We first prove the closure under restriction to a context-sensitive
  set of vertices.  Let $G$ be a linearly bounded graph defined by a
  family of $(T_a)_{a \in \Sigma}$ of $1$-incremental
  context-sensitive transduction. By Prop.~\ref{prop:bool}, the
  transduction $T_a \cap \{(u,v) \;|\; u \in F, v \in F \; \text{and}
  |v| \leq |u|+1 \}$ is $1$-incremental context-sensitive. It follows
  that $G|_{F}$ is a linearly bounded graph.

  Let us now prove the closure of linearly bounded graphs under
  restriction to reachable vertices. Let $G$ be a linearly bounded
  graph given by a family of $1$-incremental context-sensitive
  transductions $(T_a)_{a \in \Sigma}$ and $u_0$ a vertex in $V_G$.
  Each transduction $T_a$ for $a \in \Sigma$ is accepted by a LBM
  $M_a$.  Consider the graph $G'$ obtained by restricting $G$ to its
  set of vertices reachable from $u_0$. Note that this set is in
  general not context-sensitive, as illustrated by the following
  example.

  % Todo verifier la citation à Immerman
  \begin{example}
    \label{ex:restr}
    Consider a language $L \subseteq \{a,b\}^*$ that is recognizable
    in exponential space but not in linear space (i.e. $L \in
    \mathrm{NSPACE}[2^n] \setminus \mathrm{NSPACE}[n]$).  The
    existence of such a language is guaranteed by
    Thm.~\ref{thm:hierar}.  Consider the graph $G$ labeled by
    $\{a,b,\sharp\}$ and defined by:
    \begin{itemize}
    \item $\sharp u \era{x}{G} \sharp ux$ for all $u \in \{a,b\}^*$
      and $x \in \{a,b\}$,
    \item $\sharp u \sharp^n \era{\sharp}{G} \sharp u \sharp^{n+1}$
      for all $u \in \{a,b\}^*$ and $n+1 \leq 2^{|u|}$,
    \item and $\sharp u \sharp^{2^{|u|}} \era{\sharp}{G} u$ for all $u
      \in L$.
    \end{itemize}
    The edge relations of this graph are 1-incremental
    context-sensitive transductions, hence $G$ is linearly bounded.
    Figure \ref{fig:restriction} illustrates the construction of this
    graph. The set of vertices reachable from $\sharp$ restricted to
    $\{a,b\}^*$ is precisely the language $L$ which by definition is
    not context-sensitive.
  \end{example}

  \begin{figure}
    \begin{center}
      \begin{picture}(96,44)
        \gasset{Nframe=n,Nadjustdist=1,Nadjust=wh}
          
        \put(16,44){
          \node(ge)(0,0){$\sharp$}
          \node(g0)(-10,-16){$\sharp a$}
          \node(g1)(10,-16){$\sharp b$}
          \node(g1d)(30,-16){$\sharp b \sharp$}
          \node(g1dd)(50,-16){$\sharp b \sharp^2$}
          \node(g00)(-17.5,-28){}
          \node(g01)(-2.5,-28){}
          \node(g10)(0,-32){}
          \node(g11)(20,-32){$\sharp bb$}
          \node(g11d)(32,-32){$\sharp bb\sharp$}
          \node(g11dd)(44,-32){$\sharp bb\sharp^2$}
          \node(g11ddd)(56,-32){$\sharp bb\sharp^3$}
          \node(g11dddd)(68,-32){$\sharp bb\sharp^4$}
          \node(g11ddddd)(80,-32){$bb$}
          \node(g100)(-7.5,-44){}
          \node(g101)(7.5,-44){}
          \node(g110)(12.5,-44){}
          \node(g111)(27.5,-44){}

          \drawedge[ELside=r](ge,g0){$a$}
          \drawedge(ge,g1){$b$}
          \drawedge[ELside=r](g1,g10){$b$}
          \drawedge(g1,g11){$b$}

          \drawedge[curvedepth=2](g1,g1d){$\diese$}
          \drawedge[curvedepth=2](g1d,g1dd){$\diese$}
          \drawedge[curvedepth=2](g11,g11d){$\diese$}
          \drawedge[curvedepth=2](g11d,g11dd){$\diese$}
          \drawedge[curvedepth=2](g11dd,g11ddd){$\diese$}
          \drawedge[curvedepth=2](g11ddd,g11dddd){$\diese$}
          \drawedge[curvedepth=2](g11dddd,g11ddddd){$\diese$}

          \gasset{AHnb=0,dash={0.3 1}0}
          \drawedge(g0,g00){}
          \drawedge(g0,g01){}
          \drawedge(g10,g100){}
          \drawedge(g10,g101){}
          \drawedge(g11,g110){}
          \drawedge(g11,g111){}
        }
      \end{picture}
    \end{center}
    \caption{The graph $G$ for some $L$ containing $bb$ but not $b$.}
    \label{fig:restriction}
  \end{figure}
  
  Therefore, we need to adopt a different naming of the vertices of
  $G'$. We construct a graph $H$ isomorphic to $G'$ which is defined
  by a family of $1$-incremental context-sensitive transductions. For
  any pair of vertices $u$ and $v \in V_G$, we write $u \Rightarrow_n
  v$ if there exists a path from $u$ to $v$ in $G$ using only vertices
  of size less than or equal to $n$. Note that since $u$ and $v$ are
  considered as part of any path between $u$ and $v$, then necessarily
  $n \geq \max(|u|,|v|)$. The set of vertices of $H$ is $V_H=\{
  u\square^n \mid u \in \Gamma^*, u_0 \Rightarrow_{|u|+n} u \;
  \text{and}\; u_0 \not\Rightarrow_{|u|+n-1} u\}$. Note that for all
  $u \in V_{G'}$, there is a unique $n$ such that $u\square^n \in
  V_H$, and that $u \in V_G \setminus V_{G'} \implies u\square^n
  \not\in V_H$ for all $n$. For all $a \in \Sigma$, we take
  $u\square^n \era{a}{H} v\square^m$ iff $u\square^n \in V_H$,
  $v\square^m \in V_H$ and $u \era{a}{G} v$.

  The graph $H$ is isomorphic to $G'$.  It is easy to see that the
  mapping $\rho \in V_{G'} \mapsto V_H$ associating to every $u \in
  V_{G'}$ the unique $u\square^n \in V_H$ is a bijection and a graph
  morphism from $G'$ to $H$. We first prove that $V_H$ is
  context-sensitive and that for all $a \in \Sigma$, $\era{a}{H}$ is a
  $1$-incremental context-sensitive transduction. It will then follow
  that $H$ and $G'$ are linearly bounded.
  
  \begin{enumerate}
  \item Consider the language $L^+$ equal to $\{ u\square^n \;|\; u_0
    \Rightarrow_{|u|+n} u, \; n \in \mathbb{N} \}$ and the language
    $L^-$ equal to $\{ u\square^n \;|\; u_0 \Rightarrow_{|u|+n-1} u,
    \; n \in \mathbb{N}\}$, it is easy to check that $V_H=L^+ \cap
    \overline{L^-}$. If we prove that $L^+$ and $L^-$ are
    context-sensitive languages, it follows by Thm.~\ref{thm:csb} that
    $V_H$ is a context-sensitive language.
 
    We construct a LBM $M$ accepting $L^+$. When starting with
    $u\square^n$, $M$ guesses a path from $u_0$ to $u$ with vertices
    of length at most $|u|+n$.  It starts with $u_0$ and guesses a
    word $u_1$ of length at most $|u|+n$ then simulates one of the
    $M_a$'s to check that $u_0 \era{a}{G} u_1$. Finally, $M$ replaces
    $u_0$ by $u_1$ and iterates the process until $u_i$ is equal to
    $u$. This can be done using at most $3(|u|+n)$ cells. A similar
    construction allows to recognize $L^-$.
  \item For all $a \in \Sigma$, $u\square^n \era{a}{H} v\square^m$
    implies that $|v| + m \leq |u| + n + 1$, because $u_0
    \Rightarrow_{|u|+n} u$, $u \era{a}{G} v$ and $\era{a}{G}$ is
    $1$-incremental. Therefore, $\era{a}{H}$ is $1$-incremental.

    \noindent The language $L_a=\{u\square^n \sharp v\square^m \;|\; u
    \era{a}{H} v \}$ is accepted by a LBM that checks that
    $u\square^n$ and $v\square^m$ belong to $V_H$ (by simulating a
    machine accepting $V_H$) and that $u \era{a}{G} v$. Hence
    $\era{a}{H}$ is a $1$-incremental context-sensitive transduction.
    \qed
  \end{enumerate}
\end{proof}

\begin{example}
  This last construction applied to the graph of Ex. \ref{ex:restr}
  gives the graph $H$ defined by:
  \begin{itemize}
  \item $\sharp u \era{x}{G} \sharp ux$ for all $u \in \{a,b\}^*$ and
    $x \in \{a,b\}$,
  \item $\sharp u \sharp^n \era{\sharp}{G} \sharp u \sharp^{n+1}$ for
    all $u \in \{a,b\}^*$ and $n+1 \leq 2^{|u|}$,
  \item and $\sharp u \sharp^{2^{|u|}} \era{\sharp}{G}
    u\square^{2^{|u|}+1}$ for all $u \in L$.
  \end{itemize}
\end{example}

Since all rational languages are context-sensitive, linearly bounded
graphs are also closed under restriction to a rational set of
vertices. This shows that it is not necessary to allow arbitrary
rational restrictions in the definition of transition graphs of
linearly bounded machines, since such a restriction can be directly
applied to the set of external configurations of a machine.

Finally, we can extend the result of \cite{Knapik99} to the class of
linearly bounded graphs, and show that they are closed under
synchronized products \cite{Arnold82}. The synchronized product of two
graphs $G$ and $G'$ with labels in $\Sigma$ and $\Sigma'$ with respect
to a set of constraints $C \subseteq \Sigma \times \Sigma'$ is the
graph
\[
G \otimes G' = \{ (u,v) \erb{(a,b)} (u',v') \mid u \era{a}{G} u' \land
v \era{b}{G} v' \land (a,b) \in C \}.
\]

\begin{proposition}
  The class of linearly bounded graphs is closed under synchronized
  products, up to isomorphism.
\end{proposition}

\begin{proof}
  It is straightforward from this definition that if $G$ and $G'$ are
  both linearly bounded, defined for instance as context-sensitive
  transduction graphs, the binary edge relations of their synchronized
  product are also incremental context-sensitive transductions. \qed
\end{proof}

\subsection{Deterministic linearly bounded graphs}

We now consider the relations between the determinism of linearly
bounded graphs and the determinism of the linearly bounded machines
defining them. As a first remark, it is straightforward to notice that
there exist non-deter\-ministic labeled LBMs whose transition graphs are
deterministic. More precisely, any machine in which, from any given
configuration, at most one external configuration is reachable by a
partial run labeled by $\eps^*$, has a deterministic transition graph.
Figure \ref{fig:football} illustrates this fact.

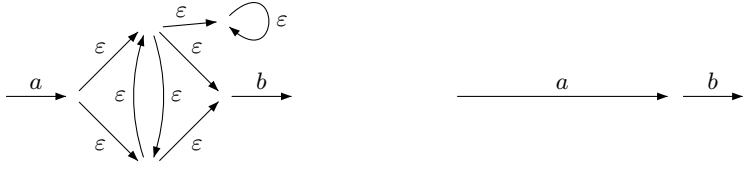
\begin{figure}
  \begin{center}
  \begin{picture}(90,25)
    \gasset{Nframe=n,Nadjust=wh,Nadjustdist=1}

    \put(0,10){
      \node(e)(0,0){}
      \node(a)(10,0){}
      \node[Nadjustdist=2](ae1)(20,10){}
      \node[Nadjustdist=2](ae2)(20,-10){}
      \node(ae3)(30,10){}
      \node(aee)(30,0){}
      \node(aeeb)(40,0){}

      \drawedge(e,a){$a$}
      \drawedge(a,ae1){$\eps$}
      \drawedge[ELside=r](a,ae2){$\eps$}
      \drawedge[sxo=-1,syo=-1](ae1,ae3){$\eps$}
      \drawedge[curvedepth=2](ae1,ae2){$\eps$}
      \drawedge[curvedepth=2](ae2,ae1){$\eps$}
      \drawloop[loopangle=0,loopdiam=5](ae3){$\eps$}
      \drawedge(ae1,aee){$\eps$}
      \drawedge[ELside=r](ae2,aee){$\eps$}
      \drawedge(aee,aeeb){$b$}
    }

    \put(60,10){
      \node(e)(0,0){}
      \node(aee)(30,0){}
      \node(aeeb)(40,0){}

      \drawedge(e,aee){$a$}
      \drawedge(aee,aeeb){$b$}
    }
  \end{picture}
  \end{center}
  \caption{Configuration graph and deterministic transition graph of a
    non-deter\-ministic labeled LBM}
  \label{fig:football}
\end{figure}

In fact, we can show that any LLBM can be transformed (while
preserving the accepted language) in order to verify this property.
Consequently, as expressed by the following proposition, all
context-sensitive languages can be accepted by a deterministic
linearly bounded graph.

\begin{proposition}
  \label{prop:trace-det}
  For every context-sensitive language $L$, there exists a
  \emph{deterministic} linearly bounded graph $G$, a vertex $i$ and a
  rational set of vertices $F$ of $G$ such that $L = L(G,\{i\},F)$.
  Moreover, $G$ can be chosen to be a tree.
\end{proposition}

\begin{proof}
  Let $L \subseteq \Sigma^*$ be a context-sensitive language, we build
  a linearly boun\-ded graph $G$ labeled by $\Sigma$ whose vertices
  are words of the form $Aw$ or $Rw$, with $w \in \Sigma^*$ and $A,R
  \not\in \Sigma$, and whose edges are defined by the following
  $1$-incremental context-sensitive transductions:
  \begin{align*}
    Au \erb{a} Aua, & \quad Rv \erb{a} Ava \quad && \text{for all } u
    \in L, v \not\in L \;\text{and}\; ua, va \in L,
    \\
    Au \erb{a} Rua, & \quad Rv \erb{a} Rva \quad && \text{for all } u
    \in L, v \not\in L \;\text{and}\; ua, va \not \in L.
  \end{align*}
  The language of this deterministic graph between vertex $X$
  where $X$ is equal to $A$ if $\eps \in L$ or equal to $R$ if $\eps
  \not\in L$ and the rational set $A\Gamma^*$ is precisely $L$.
  Moreover, $G$ is a tree isomorphic to the complete infinite
  $\Sigma$-labeled tree. \qed
\end{proof}

\begin{remark}
  The previous result does not stand for a finite set of final
  vertices even if we only consider deterministic context-sensitive
  languages. Consider for example, the language $L=\{(a^nb)^* | n \in
  \mathbb{N} \}$. Suppose that there exists a deterministic graph $G$
  such that $L=L(G,i,F)$ for some finite set $F$, then there would
  exists $m \not= n$, such that $i \era{a^nb}{G} f$ and $i
  \era{a^mb}{G} f$. As $f \era{a^mb}{G} f'$ for some $f' \in F$, it
  would follow that $i \era{a^nba^mb}{G} f'$.
\end{remark}

Of course, we cannot conclude from this that the languages of
deterministic linearly bounded graphs are the deterministic
context-sensi\-tive languages, which would amount to answering the
general question raised by Kuroda \cite{Kuroda64} whether
deterministic and non-deterministic languages coincide. However, if we
only consider terminating linearly bounded machines (which have no
infinite run on any given input word) the class of transition graphs
we obtain faithfully illustrates the determinism of the languages.

\begin{remark}
  Even for terminating LLBMs, the determinism of a transition
  graph does not imply that the machine defining it is deterministic.
  Figure \ref{fig:qrt} illustrates this fact. The idea of the
  following proof is to show that any terminating LLBM whose
  transition graph is deterministic can be determinized without
  changing the structure of the graph.
\end{remark}

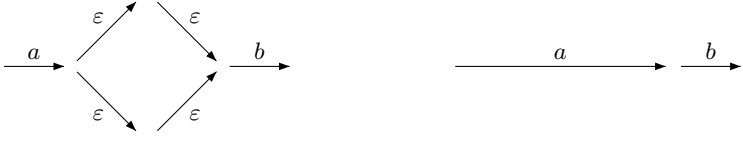
\begin{figure}
  \begin{center}
    \begin{picture}(90,25)
      \gasset{Nframe=n,Nadjust=wh,Nadjustdist=1}
      
      \put(0,10){
        \node(e)(0,0){}
        \node(a)(10,0){}
        \node[Nadjustdist=2](ae1)(20,10){}
        \node[Nadjustdist=2](ae2)(20,-10){}
        \node(aee)(30,0){}
        \node(aeeb)(40,0){}
        
        \drawedge(e,a){$a$}
        \drawedge(a,ae1){$\eps$}
        \drawedge[ELside=r](a,ae2){$\eps$}
        \drawedge(ae1,aee){$\eps$}
        \drawedge[ELside=r](ae2,aee){$\eps$}
        \drawedge(aee,aeeb){$b$}
      }
      
      \put(60,10){
        \node(e)(0,0){}
        \node(aee)(30,0){}
        \node(aeeb)(40,0){}
        
        \drawedge(e,aee){$a$}
        \drawedge(aee,aeeb){$b$}
      }
    \end{picture}
  \end{center}
  \caption{Configuration graph and deterministic transition graph of a
    non-deter\-ministic terminating labeled LBM}
  \label{fig:qrt}
\end{figure}

\begin{proposition}
\label{prop:det}
The languages of deterministic transition graphs of terminating
linearly bounded machines (between an initial vertex and a rational
set of final vertices) are the deterministic context-sensitive
languages.
\end{proposition}

% généraliser aux machines de Turing qui s'arrêtent ?

\begin{proof}
  Let $L$ be a deterministic context-sensitive language. By Prop.
  \ref{prop:eqlbmllbm}, there exists a deterministic terminating LLBM
  $M$ accepting $L$. By definition, its transition graph $G_M$ is
  deterministic, and it accepts $L$ between the vertex corresponding
  to the initial configuration and those corresponding to the final
  configurations of $M$.
  
  Conversely, let $G_M$ be the deterministic transition graph of a
  terminating LLBM $M=(\Gamma,\Sigma,[,],Q,q_0,F,\delta)$, we
  construct a deterministic LLBM $N$ whose transition graph $G_{N}$ is
  equal to $G_M$. The machine $N$ is obtained from $M$ by keeping
  exactly one rule in $\delta$ with a given left-hand side and label.
  The machine $N$ is equal to $(\Gamma,\Sigma,[,],Q,q_0,F,\delta')$
  where $\delta'$ is a subset of $\delta$:
  \begin{itemize}
  \item if $pA \erb{x} qU \in \delta$ then there exists $pA \erb{x}
    q'U' \in \delta'$,
  \item if $pA \erb{x} qU \in \delta'$ then for all $pA \erb{x} q'U'
    \in \delta'$, $q=q'$ and $U=U'$.
  \end{itemize}
  By construction, $N$ is deterministic. Moreover, the set of external
  configurations of $N$ is equal to the set of external configurations
  of $M$. It remains to prove that for any two external configurations
  $c$ and $c'$, $c \era{a}{G_M} c'$ if and only if $c \era{a}{G_{N}}
  c'$.

  If $c \era{a}{G_M} c'$, there exists a path between $c$ and $c'$
  labeled by $a \varepsilon^*$ in the configuration graph $C_M$ of
  $M$.  As $G_M$ is deterministic and terminating, any maximal path in
  $C_M$ labeled by $a \varepsilon^*$ ends in $c'$. This property still
  holds for the configuration graphs of $C_{N}$. Suppose by
  contradiction that there exists a maximal path in $C_{N}$ labeled by
  $a\varepsilon^*$ starting from $c$ and ending in $c'' \not=c'$.
  This implies the $c''$ has out-going $\epsilon$-edges in $C_M$ and
  not in $C_{N}$ which is impossible by construction of $N$.
  
  Conversely if $c \era{a}{G_N} c'$ then, by construction of $N$, we
  have $c \era{a}{G_M} c'$. Hence, it follows that $G_M=G_N$. \qed
\end{proof}

\begin{remark}
  Other equivalent definitions of deterministic transition graphs of
  terminating labeled LBMs can be given: they coincide with the
  Cayley-type graphs of length-decreasing rewriting systems with
  unique normal forms\footnote{A rewriting system has unique normal
    forms if any word can be derived into at most one normal form.},
  and to the graphs of deterministic terminating incremental
  context-sensitive transductions. The intuitive reason behind this
  equivalence is that, in every such case, for each label the
  corresponding edge relation of the graph is a partial function from
  vertices to vertices.
\end{remark}

\begin{remark}
  One can not use the previous constructions to determinize any LLBM.
  Indeed, the construction from Prop. \ref{prop:trace-det} produces
  non terminating machines in general, and the construction from Prop.
  \ref{prop:termlbm} does not preserve the structure of a machine's
  transition graph.
\end{remark}

\subsection{Logical properties}

To conclude this section on properties of linearly bounded graphs, we
investigate the decidability of the first-order theory of
configuration graphs of LLBMs and linearly bounded graphs. Due to the
high expressive power of the model considered, only local properties
expressed in first-order logic can be checked on LLBM configuration
graphs.

\begin{proposition}
  Configuration graphs of linearly bounded Turing machines have a
  decidable first-order theory.
\end{proposition}

\begin{proof}
  This is a direct consequence of the fact that configuration graphs
  of LLBMs are synchronized rational graphs, which have a decidable
  first-order theory \cite{Blumensath00,KhoussainovN94}. \qed
\end{proof}

However, as remarked by \cite{Knapik99}, there exists no algorithm
which, given a LLBM $M$ and a first-order sentence $\phi$, decides
whether the transition graph of $M$ satisfies $\phi$. This statement
can be strengthened to the following proposition.

\begin{proposition}
  \label{prop:lbgfo}
  There exists a linearly bounded graph with an undecidable
  first-order theory.
\end{proposition}

\begin{proof}
  Consider a fixed enumeration $(M_n)_{n \in \mathbb{N}}$ of all
  (unlabeled) Turing machines, and a labeled Turing machine $M$ whose
  language is the set of all words of the form $\#^n$ such that
  machine $M_n$ halts on the empty input.  Using only
  $\eps$-transitions, $M$ guesses a number $n$ and writes $\#^n$ on
  its tape, then simulates machine number $n$ on the empty input. If
  the machine halts, then $M$ reads word $\#^n$ and stops.  All
  accepting runs of $M$ are thus labeled by $\eps^* \#^n$ for some
  $n$.
  
  When replacing $\eps$ by an observable symbol $\tau$, the
  configuration graph $C$ of $M$ is a linearly bounded graph. Let $C'$
  be its restriction to vertices reachable from the initial
  configuration of $M$, the graph $C'$ is still linearly bounded by
  Prop.  \ref{prop:restriction}.  For all $n$, the formula $\phi_n$
  expressing the existence of a path labeled by $\tau \#^n$ ending in
  a vertex with no successor is satisfied in $C'$ if and only if
  machine $M_n$ halts on input $\eps$. The set of all such satisfiable
  formulas is not recursive. Hence the first-order theory of $C'$ is
  undecidable (not even recursively enumerable). \qed
\end{proof}

\section{Comparison with rational graphs}
\label{sect:comp}

We will now give some remarks about the comparison between linearly
boun\-ded graphs and several different sub-classes of rational
graphs. First note that since linearly boun\-ded graphs have by
definition a finite degree, it is therefore only relevant to consider
rational graphs of finite degree. However, even under this structural
restriction, rational and linearly-bounded graphs are incomparable,
due to the incompatibility in the growth rate of their vertices'
degrees. 

A first observation is that in a rational graph the out-degree at
distance $n$ from any vertex can be $c^{c^n}$ for some constant $c$,
whereas in a linearly bounded graph it is at most $c^n$.

\begin{lemma}
  \label{lem:lbg-od}
  For any linearly bounded graph $G$ and any vertex $x$, there exists
  $c \in \mathbb{N}$ such that the out-degree of $G$ at distance $n>0$
  of $x$ is at most $c^n$.
\end{lemma}

\begin{proof}
  Let $(T_a)_{a \in \Sigma}$ be a set of incremental context-sensitive
  transductions describing $G$ and $k_a \in \mathbb{N}$ such that
  $T_a$ is a $k_a$-incremental context-sensitive transduction.  We
  take $k$ to be the maximum of $\{ k_a \; | \; a \in \Sigma \} \cup
  \{ |x| \}$. At distance $n>0$ of $x$, a vertex has length at most
  $k(n+1)$ and hence the out-degree is bounded by the number of
  vertices of length at most $k(n+2)$ which is less than
  $|\Gamma|^{k(n+2)+1}$. Hence, there exists $c \in \mathbb{N}$ such
  that the out-degree is bounded by $c^n$. \qed
\end{proof}

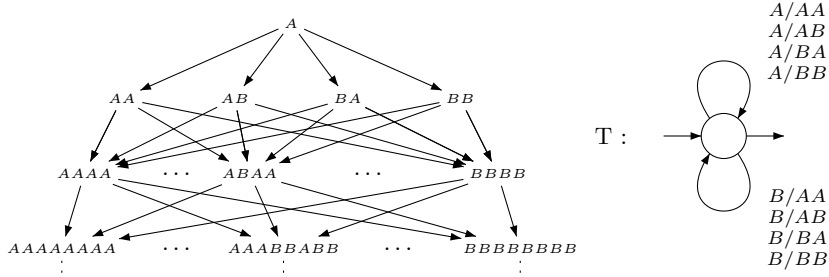
\begin{figure}
  \begin{center}
    \begin{picture}(110,33)(0,0)
      \put(95,15){
        \node[Nframe=n](T)(-15,0){T :}
        \node[Nmarks=if,Nw=6,Nh=6](q)(0,0){}

        \drawloop[ELpos=70,loopdiam=7](q){\scriptsize 
          $\begin{array}{c}A / AA \\ A / AB \\ A / BA \\ A / BB\end{array}$}
        \drawloop[ELpos=30,loopdiam=7,loopangle=-90](q){\scriptsize 
          $\begin{array}{c}B / AA \\ B / AB \\ B / BA \\ B / BB\end{array}$}
      }

      \gasset{Nframe=n,Nadjust=wh,Nadjustdist=1}

      \node(A)(37.5,30){\tiny $A$}
      \node(AA)(15,20){\tiny $AA$}
      \node(AB)(30,20){\tiny $AB$}
      \node(BA)(45,20){\tiny $BA$}
      \node(BB)(60,20){\tiny $BB$}
      \node(AAAA)(10,10){\tiny $AAAA$}
      \node(dots1)(22.5,10){\dots}
      \node(ABAA)(32,10){\tiny $ABAA$}
      \node(dots2)(48,10){\dots}
      \node(BBBB)(65,10){\tiny $BBBB$}
      \node(AAAAAAAA)(7,0){\tiny $AAAAAAAA$}
      \node(ddoottss1)(22.5,0){\dots}
      \node(AAABBABB)(36.5,0){\tiny $AAABBABB$}
      \node(ddoottss2)(52,0){\dots}
      \node(BBBBBBBB)(68,0){\tiny $BBBBBBBB$}
      \node(AAAAAAAA-)(7,-5){}
      \node(AAABBABB-)(36.5,-5){}
      \node(BBBBBBBB-)(68,-5){}

      \drawedge(A,AA){}
      \drawedge(A,AB){}
      \drawedge(A,BA){}
      \drawedge(A,BB){}
      \drawedge(AA,AAAA){}
      \drawedge(AA,ABAA){}
      \drawedge(AA,BBBB){}
      \drawedge(AA,AAAA){}
      \drawedge(AB,ABAA){}
      \drawedge(AB,BBBB){}
      \drawedge(AB,AAAA){}
      \drawedge(AB,ABAA){}
      \drawedge(BA,BBBB){}
      \drawedge(BA,AAAA){}
      \drawedge(BA,ABAA){}
      \drawedge(BA,BBBB){}
      \drawedge(BB,BBBB){}
      \drawedge(BB,AAAA){}
      \drawedge(BB,ABAA){}
      \drawedge(BB,BBBB){}
      \drawedge(AAAA,AAAAAAAA){}
      \drawedge(AAAA,AAABBABB){}
      \drawedge(AAAA,BBBBBBBB){}
      \drawedge(ABAA,AAAAAAAA){}
      \drawedge(ABAA,AAABBABB){}
      \drawedge(ABAA,BBBBBBBB){}
      \drawedge(BBBB,AAAAAAAA){}
      \drawedge(BBBB,AAABBABB){}
      \drawedge(BBBB,BBBBBBBB){}

      \gasset{AHnb=0,dash={0.3 1}0}
      \drawedge(AAAAAAAA,AAAAAAAA-){}
      \drawedge(AAABBABB,AAABBABB-){}
      \drawedge(BBBBBBBB,BBBBBBBB-){}
    \end{picture}
  \end{center}
  \caption{A finite degree rational graph (together with its
    transducer) which is not isomorphic to any linearly bounded
    graph.}
  \label{fig:triplex}
\end{figure}

Figure \ref{fig:triplex} shows a rational graph whose vertices at
distance $n$ from the root $A$ have out-degree $2^{2^{n+1}}$. This
graph is thus not linearly bounded. 

Conversely, in a rational graph of finite degree, the in-degree at
distance $n$ from any vertex is at most $c^{c^n}$ for some $c \in
\mathbb{N}$, in a linearly bounded graph it can be as large as $f(n)$
for any mapping $f$ from $\mathbb{N}$ to $\mathbb{N}$ recognizable in
linear space (i.e. such that the language $\{ 0^n 1^{f(n)} \; | \; n
\in \mathbb{N}\}$ is context-sensitive).

\begin{lemma}
\label{lem:lbg-id}
  For any mapping $f: \mathbb{N} \mapsto \mathbb{N}$ recognizable in
  linear space, there exists a linearly bounded graph $G$ with a
  vertex $x$ such that the in-degree at distance $n>0$ of $x$ is
  $f(n)$.
\end{lemma}

\begin{proof}
  Let $f$ be a mapping from $\mathbb{N}$ to $\mathbb{N}$ recognizable
  in linear space, and let $G_f$ be the linearly bounded graph defined
  using the following incremental context-sensitive transduction:
  \begin{gather*}
    T\ =\ \{ (u,u0) \sep u \in 0^* \}\ \cup\ \{ (u,u1) \sep u \in
    0^n1^m,\ m < f(n) \}\\ \cup\ \{ (uv,u) \sep u \in 0^*, v \in 1^*
    \;\text{and}\; |v| \leq f(|u|)\}.
  \end{gather*}
  Note that, since context-sensitive languages are closed under
  complement, it is not difficult to see that the set $\{0^n1^m \mid m
  < f(n) \}$ is context-sensitive for every $f$ recognizable in linear
  space. Figure \ref{fig:construction} illustrates the construction of
  $G_f$.  The in-degree of vertex $0^n$ at distance $n$ from the root
  $\eps$ is equal to $f(n)$, for any mapping $f$ recognizable in
  linear space. \qed

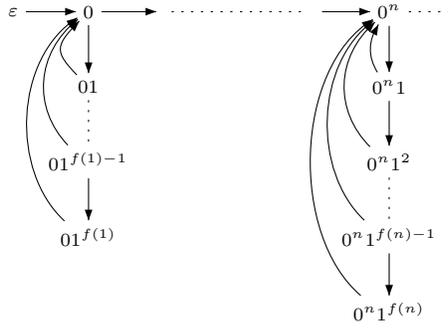
\begin{figure}
  \begin{center}
    \begin{picture}(60,45)(0,0)

      \gasset{Nframe=n,Nadjust=wh,Nadjustdist=1}
      \node(e)(0,40){\scriptsize $\eps$}
      \node(0)(10,40){\scriptsize $0$}
      \node(00)(20,40){}
      \node(0n-1)(40,40){}
      \node(0n)(50,40){\scriptsize $0^n$}
      \node(etc)(60,40){}
      
      \node(01)(10,30){\scriptsize $01$}
      \node(011)(10,20){\scriptsize $01^{f(1)-1}$}
      \node(01f)(10,10){\scriptsize $01^{f(1)}$}
      
      \node(0n1)(50,30){\scriptsize $0^n1$}
      \node(0n11)(50,20){\scriptsize $0^n1^2$}
      \node(0n1f-1)(50,10){\scriptsize $0^n1^{f(n)-1}$}
      \node(0n1f)(50,0){\scriptsize $0^n1^{f(n)}$}

      \drawedge(e,0){}
      \drawedge(0,00){}
      \drawedge[AHnb=0,dash={0.3 1}0](00,0n-1){}
      \drawedge(0n-1,0n){}
      \drawedge[AHnb=0,dash={0.3 1}0](0n,etc){}

      \drawedge(0,01){}
      \drawedge[AHnb=0,dash={0.3 1}0](01,011){}
      \drawedge(011,01f){}

      \drawedge(0n,0n1){}
      \drawedge(0n1,0n11){}
      \drawedge[AHnb=0,dash={0.3 1}0](0n11,0n1f-1){}
      \drawedge(0n1f-1,0n1f){}

      \drawbcedge(01,5,35,0,5,35){}
      \drawbcedge(011,2,25,0,2,35){}
      \drawbcedge(01f,-1,15,0,-1,35){}

      \drawqbedge(0n1,45,35,0n){}
      \drawbcedge(0n11,42,25,0n,42,35){}
      \drawbcedge(0n1f-1,39,15,0n,39,35){}
      \drawbcedge(0n1f,36,5,0n,36,35){}
    \end{picture}
  \end{center}
  \caption{A linearly bounded graph which is not isomorphic to any
    rational graph.}
  \label{fig:construction}
\end{figure}

\end{proof}

An instance of such a mapping is $f: n \mapsto 2^{2^{2^n}}$, which is
more than the in-degree at distance $n$ of a vertex in any rational
graph of finite degree. From these two observations, we easily get:

\begin{proposition}
  The classes of finite degree rational graphs and of linearly
  boun\-ded graphs are incomparable.
\end{proposition}

Since finite-degree rational graphs and linearly bounded graphs are
incomparable, we investigate more restricted sub-classes of rational
graphs. For synchronized graphs of finite out-degree, we have the
following result:

\begin{proposition}
  The synchronized graphs of finite out-degree form a strict sub-class
  of linearly bounded graphs.
\end{proposition}

\begin{proof}
  By Prop.~\ref{prop:synchronizedincremental}, the synchronized
  relations of finite image are $k$-in\-cre\-mental. Hence, the
  synchronized graphs of finite out-degree are linearly boun\-ded
  graphs. The strictness can be deduced from the fact that
  synchronized graphs are not closed under restriction to reachable
  vertices from a given vertex (cf Prop.~\ref{prop:restriction}), or
  that synchronized graphs have decidable first-order theories (cf
  Prop. \ref{prop:lbgfo}). \qed
\end{proof}

By restricting the out-degree even further, we can establish the
following comparison:

\begin{theorem}
  \label{thm:bd-rat-lbg}
  The rational graphs of bounded out-degree form a sub-class of
  linearly bounded graphs (up to isomorphism).
\end{theorem}

\begin{proof}
  The inclusion result is based on a uniformization result for
  rational relations of bounded image due to Weber.  A relation is
  said to be $k$-valued if for all $w \in \textrm{Dom}(R)$, we have
  $|\{ x \; | \; (w,x) \in R \}| \leq k$.  In \cite{Weber96}, it is
  proved that for any $k$-valued rational relation $R$, there exist
  $k$ rational functions\footnote{A rational relation which maps
    exactly one image to each element of its domain is called a
    rational function.} $F_1, \dots F_k$ such that $R = \bigcup_{i \in
    [1,k]} F_i$.

  Let $G=(T_a)_{a \in \Sigma}$ be a rational graph over $\Gamma$ whose
  out-degree is bounded by $k$. Each relation $T_a$ is therefore
  $k$-valued and hence there exist $k$ rational functions $F_{a_1},
  \ldots F_{a_k}$ such that $T_a = \bigcup_{i \in [1,k]} F_{a_i}$. We
  write $X=\{ a_i \; | \; a \in \Sigma \; \textrm{and} \; i \in [1,k]
  \}$.

  We will represent each vertex $x$ of $G$ by a pair $(w,t) \in V_G
  \times X^*$ such that $x=F_{t_{|t|}}(\ldots(F_{t_1}(w)))$.  However,
  there can be several such pairs for a given vertex. We define a
  total order $<$ on $V_G \times X^*$ and associate to $x$ the
  smallest such pair. First, let $<_\Gamma$ and $<_X$ be two arbitrary
  total orders over $\Gamma$ and $X$ respectively and let
  $\prec_\Gamma$ and $\prec_X$ be the lexicographic orders induced by
  $<_\Gamma$ and $<_X$ respectively.
  \[
  \begin{array}{lcll}
    (m,t) < (n,r) &  \quad \textrm{iff} \quad &     &|m|+|t| < |n|+|r| 
    \\
    &     &   \quad  \textrm{or} \quad& |m|+|t|=|n|+|r|
    \;\textrm{and}\; m \prec_\Gamma n 
    \\
    &      &   \quad \textrm{or} & |m|+|t|=|n|+|r|,  m =  n \;
    \textrm{and} \; t \prec_X r 
  \end{array}
  \]
  It is straightforward to prove that $<$ is a total order on $V_G
  \times X^*$.  We now prove that the function $N$ associating to any
  pair $(w,t)$ the smallest pair $(m,r)$ such that $F_t(w) = F_r(m)$
  is a 0-incremental context-sensitive transduction.
  
  There are two important points in this proof. The first one is to be
  able to check in linear space that for any two given pairs $(m,r)$
  and $(w,t)$, $F_r(m) = F_t(w)$. In other terms, we want to prove
  that the language $L=\{ m\sharp r \$ t \sharp w \; | \; F_r(m) =
  F_t(w)\}$ is context-sensitive.

  \noindent
  Let $\bar{X}$ be a finite alphabet disjoint from but in bijection
  with $X$ (for all $x \in X$, we write $\bar{x}$ for the
  corresponding symbol in $\bar{X}$) , we consider the rational graph
  $\bar{G}$ defined by the family of transducers $(F_a)_{a \in X} \cup
  (\bar{F}_{\bar{a}})_{a \in X}$ where for all $\bar{a} \in \bar{X}$,
  $\bar{F}_{\bar{a}} = F^{-1}_a$. It is easy to check that for all
  $m,w \in \Gamma^*$ and all $r,t \in X^*$, there exists a path $m
  \eRa{r\bar{t}}{\bar{G}} w$ iff $F_r(m)=F_t(w)$ where $\bar{t}=
  \bar{t}_{|t|} \ldots \bar{t}_1$. By Cor.~\ref{cor:traces}, the
  language $\{ i \sharp w \sharp f \; | \; w \in L(\bar{G},i,f)\} \cap
  \Gamma^*\sharp X^* \bar{X}^* \sharp \Gamma^*$ is a context-sensitive
  language. It follows that $L$ is also context-sensitive.

  The second point is to show that the language $L'=\{ m\sharp r \$ t
  \sharp w$ $\mid$ $N((w,t))$ $=$ $(m,r)\}$ is itself
  context-sensitive, which implies that $N$ is indeed a
  context-sensitive transduction.  Given $(w,t)$ and $(m,r)$, let
  $(m_1,r_1), \ldots, (m_k,r_k)$ be an enumeration of all pairs
  smaller than $(m,r)$ with respect to $<$. A linearly bounded machine
  accepting $L'$ successively tests for all $i \in [1,k]$ whether $m_i
  \sharp r_i \$ t \sharp w \not\in L$, which is possible since $L$ is
  context-sensitive and context-sensitive languages are closed under
  complement. If any of these $k$ tests fails, then the whole
  computation fails. Otherwise, it only remains to check that $m_i
  \sharp r_i \$ t \sharp w \in L$.

  $N$ is $0$-incremental since by definition of the total order $<$,
  $(m,r) < (w,t) \implies |m|+|r| \leq |w|+|t|$, and all steps of the
  construction of the machine accepting $L'$ can be done in linear
  space.

%   The second point is to perform this test for all pairs $(m,r) <
%   (w,t)$. Even though the number of such pairs is finite, there is a
%   problem when the testing fails. Indeed, this may mean that $m \sharp
%   r\$t \sharp w \not\in L$, or that the path actually exists and that
%   there exists another successful computation. Hence, on a failed test
%   we cannot simply move to the next $(m,r)$ candidate. To deal with
%   this problem, an idea is to perform a test of membership in the
%   complement $\overline{L}$ of $L$, which is a context-sensitive
%   language too. In case the complement test fails, this computation of
%   $N(w,t)$ can safely fail, since we can assure that another
%   successful computation for either the membership test or its
%   complement exists. In case the complement test succeeds, it means
%   the current pair $(m,r)$ is not valid and we can proceed to the
%   following one. The computation of $N(w,t)$ succeeds when a suitable
%   $(m,r)$ pair is found, and fails when all pairs $(m,r) < (w,t)$ have
%   been tested and rejected. $N$ is $0$-incremental since by definition
%   of the total order $<$, $(m,r) < (w,t) \implies |m|+|r| \leq
%   |w|+|t|$, and all steps of the construction can be done in linear
%   space.
    
  We can now define the linearly bounded graph $G'$ with vertices in
  $\Gamma^* \diese X^*$ using the set of transductions $(T_a)_{a \in
    \Sigma}$ with:
  \[
  T_a \; = \; \bigcup_{i \in [1,k]} \; \{(w\diese x, w' \diese y) \; |
  \; N((w,x))=(w,x) \; \textrm{and}\; N(w,x a_i)=(w',y) \}.
  \]
  As $N$ is 0-incremental, $T_a$ is 1-incremental. The mapping $\rho$
  from $V_G$ to $V_G \times X^*$ associating to a vertex $x \in V_G$
  the smallest $(w,t)$ such that $x=F_t(w)$ is a bijection from $V_G$
  to $V_{G'}$, which induces an isomorphism from $G$ to $G'$.  Hence,
  $G$ is a linearly bounded graph, which concludes the proof of
  inclusion. \qed
\end{proof}

\begin{example}
  Let $\Gamma$ be a singleton alphabet, and consider the pair of
  transductions $T_a = \{(n,2n) | n \geq 1\}$ and $T_b = \{(n,n-3)|n
  \geq 4\}$ on the domain of words over $\Gamma$ seen as unary-coded
  positive integers. Let us apply the construction from the previous
  inclusion proof to the bounded-degree rational graph $G$ defined by
  $T_a$ and $T_b$. Here, $<_\Gamma$ is trivial, and we take $a <_X b$.
  The graph $G$ and the linearly bounded graph $G'$ obtained by
  applying the construction are shown on Fig. \ref{fig:rat2lbm}. Note
  for instance how vertex $13$ in $G$ is represented by vertex
  $(2,a³b)$ in $G'$ because $(2,a³b)$ is the smallest pair $(u,v)$
  such that $u \era{v}{G} 13$.
\end{example}

\begin{figure}
  \begin{center}
    \begin{picture}(68,129.6)
      \gasset{Nframe=n,Nadjust=wh,Nadjustdist=1}
      
      \put(0,125){
        \node(G)(12,-36){\large $G$}

        \node(1)(0,0){1}
        
        \node(2)(12,0){2}
        
        \node(4)(24,0){4}

        \node(8)(36,0){8}
        \node(5)(36,-32){5}
        
        \node(16)(48,0){16}
        \node(13)(48,-16){13}
        \node(10)(48,-32){10}
        \node(7)(48,-48){7}

        \node(32)(60,0){32}
        \node(29)(60,-8){29}
        \node(26)(60,-16){26}
        \node(23)(60,-24){23}
        \node(20)(60,-32){20}
        \node(17)(60,-40){17}
        \node(14)(60,-48){14}
        \node(11)(60,-56){11}

        \drawedge(1,2){\footnotesize $a$}
        \drawedge(2,4){\footnotesize $a$}
        \drawedge(4,8){\footnotesize $a$}
        \drawedge(8,16){\footnotesize $a$}
        \drawedge[ELpos=55](5,10){\footnotesize $a$}
        \drawedge(16,32){\footnotesize $a$}
        \drawedge(13,26){\footnotesize $a$}
        \drawedge(10,20){\footnotesize $a$}
        \drawedge(7,14){\footnotesize $a$}

        \drawedge[curvedepth=5](4,1){\footnotesize $b$}
        \drawedge[ELside=r](8,5){\footnotesize $b$}
        \drawedge[curvedepth=3](5,2){\footnotesize $b$}
        \drawedge(16,13){\footnotesize $b$}
        \drawedge(13,10){\footnotesize $b$}
        \drawedge(10,7){\footnotesize $b$}
        \drawqbedge[ELpos=40](7,24,-40,4){\footnotesize $b$}
        \drawedge(32,29){}
        \drawedge(29,26){}
        \drawedge(26,23){}
        \drawedge(23,20){}
        \drawedge(20,17){}
        \drawedge(17,14){}
        \drawedge(14,11){}
        \drawbcedge[ELpos=20](11,26,-62,8,48,0){\footnotesize $b$}

        \drawline[AHnb=0,dash={0.3 1}0](63,0)(68,0)
        \drawline[AHnb=0,dash={0.3 1}0](63,-28)(68,-28)
        \drawline[AHnb=0,dash={0.3 1}0](63,-57)(68,-58)
      }

      \put(0,58){
        \node(G)(12,-36){\large $G'$}

        \node(1)(0,0){\footnotesize $1,\eps$}

        \node(2)(12,0){\footnotesize $2,\eps$}

        \node(4)(24,0){\footnotesize $2,a$}

        \node(8)(36,0){\footnotesize $2,a²$}
        \node(5)(36,-32){\footnotesize $5,\eps$}

        \node(16)(48,0){\footnotesize $2,a³$}
        \node(13)(48,-16){\footnotesize $2,a³b$}
        \node(10)(48,-32){\footnotesize $5,a$}
        \node(7)(48,-48){\footnotesize $7,\eps$}

        \node(32)(60,0){\footnotesize $2,a^4$}
        \node(29)(60,-8){\footnotesize $2,a^4b$}
        \node(26)(60,-16){\footnotesize $2,a^3ba$}
        \node(23)(60,-24){\footnotesize $2,a^3bab$}
        \node(20)(60,-32){\footnotesize $5,a²$}
        \node(17)(60,-40){\footnotesize $5,a²b$}
        \node(14)(60,-48){\footnotesize $7,a$}
        \node(11)(60,-56){\footnotesize $7,ab$}

        \drawedge(1,2){\footnotesize $a$}
        \drawedge(2,4){\footnotesize $a$}
        \drawedge(4,8){\footnotesize $a$}
        \drawedge(8,16){\footnotesize $a$}
        \drawedge[ELpos=55](5,10){\footnotesize $a$}
        \drawedge(16,32){\footnotesize $a$}
        \drawedge(13,26){\footnotesize $a$}
        \drawedge(10,20){\footnotesize $a$}
        \drawedge(7,14){\footnotesize $a$}

        \drawedge[curvedepth=5](4,1){\footnotesize $b$}
        \drawedge[ELside=r](8,5){\footnotesize $b$}
        \drawedge[curvedepth=3](5,2){\footnotesize $b$}
        \drawedge(16,13){\footnotesize $b$}
        \drawedge(13,10){\footnotesize $b$}
        \drawedge(10,7){\footnotesize $b$}
        \drawqbedge[ELpos=40](7,24,-40,4){\footnotesize $b$}
        \drawedge(32,29){}
        \drawedge(29,26){}
        \drawedge(26,23){}
        \drawedge(23,20){}
        \drawedge(20,17){}
        \drawedge(17,14){}
        \drawedge(14,11){}
        \drawbcedge[ELpos=20](11,26,-62,8,48,0){\footnotesize $b$}

        \drawline[AHnb=0,dash={0.3 1}0](63,0)(68,0)
        \drawline[AHnb=0,dash={0.3 1}0](63,-28)(68,-28)
        \drawline[AHnb=0,dash={0.3 1}0](63,-57)(68,-58)
      }
   \end{picture}
  \end{center}
  \caption{Bounded-degree rational graph $G$ and isomorphic lin.
    bounded graph $G'$.}
  \label{fig:rat2lbm}
\end{figure}
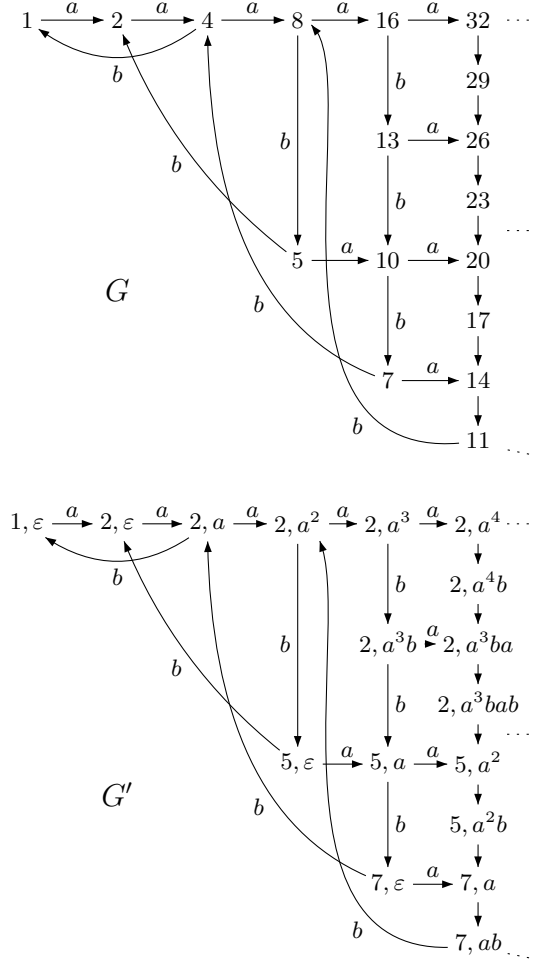

In fact, we can prove that the previous inclusion is strict.

\begin{theorem}
  There exists a linearly bounded graph of bounded degree which is not a
  rational graph.
\end{theorem}

\begin{proof}
  We now prove that there exists a linearly bounded graph of bounded
  degree which is not isomorphic to any bounded degree rational graph.
  Let us first establish that linearly bounded graphs of bounded
  degree are not closed under edge reversal.  Let $L \subseteq
  \{0,1\}^+$ be a language in $\textrm{NSPACE}[2^n] \setminus
  \textrm{NSPACE}[n]$ (cf Thm.~\ref{thm:hierar}), we define a linearly
  bounded graph $G$ with vertices in $0\{0,1\}^*\diese^* \cup
  \bar{0}\{\bar{0},\bar{1} \}^*$ and edges defined by:
  \begin{align*}          
    \erb{x} & =\{ (w,wx) \; | \; w \in 0\{ 0, 1\}^* \} \cup \{
    (w,w\bar{x}) \; | \; w \in \bar{0}\{\bar{0},\bar{1}\}^*\} \qquad
    \textrm{for $x \in \{0,1\}$}
    \\
    \erb{\diese} & = \{ (w,w\diese) \; | \; w=u v, u \in 0\{0,1\}^*, v
    \in \diese^* \;\textrm{and} \; |v| < 2^{|u|} \}
    \\
    & \quad \cup \{ (w,\bar{u}) \; | \; w=u v, u \in 0\{0,1\}^*, v \in
    \diese^*, |v| = 2^{|u|} \;\textrm{and} \; u \in L \}
  \end{align*}
  The relations $\erb{0}, \erb{1}$ and $\erb{\diese}$ are
  incremental context-sensitive transductions. In fact, as $L$ belongs
  to $\textrm{NSPACE}[2^n]$, the language $\{ w\diese^{2^{|w|}} \; | \;
  w \in L \}$ belongs to $\textrm{NSPACE}[n]$.  The construction of
  $G_L$ is illustrated by Figure \ref{fig:bitree.eps}.

  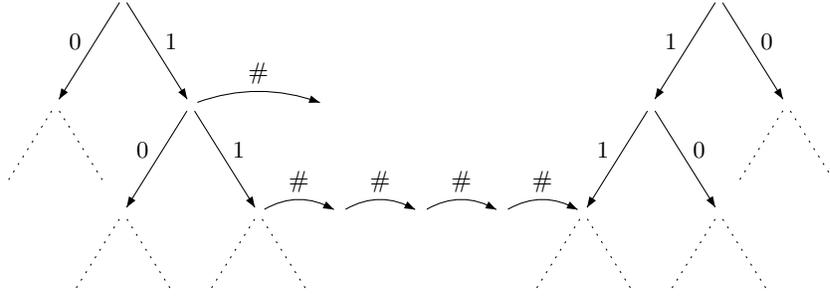
\begin{figure}
    \begin{center}
      \unitlength=0.9mm
      \begin{picture}(120,44)
        \gasset{Nframe=n,Nadjustdist=1,Nadjust=wh}
        
        \put(16,44){
          \node(ge)(0,0){}
          \node(g0)(-10,-16){}
          \node(g1)(10,-16){}
          \node(g1d)(30,-16){}
          \node(g00)(-17.5,-28){}
          \node(g01)(-2.5,-28){}
          \node(g10)(0,-32){}
          \node(g11)(20,-32){}
          \node(g11d)(32,-32){}
          \node(g11dd)(44,-32){}
          \node(g11ddd)(56,-32){}
          \node(g11dddd)(68,-32){}
          \node(g100)(-7.5,-44){}
          \node(g101)(7.5,-44){}
          \node(g110)(12.5,-44){}
          \node(g111)(27.5,-44){}

          \drawedge[ELside=r](ge,g0){0}
          \drawedge(ge,g1){1}
          \drawedge[ELside=r](g1,g10){0}
          \drawedge(g1,g11){1}

          \drawedge[curvedepth=2](g1,g1d){$\diese$}
          \drawedge[curvedepth=2](g11,g11d){$\diese$}
          \drawedge[curvedepth=2](g11d,g11dd){$\diese$}
          \drawedge[curvedepth=2](g11dd,g11ddd){$\diese$}
          \drawedge[curvedepth=2](g11ddd,g11dddd){$\diese$}

          \gasset{AHnb=0,dash={0.3 1}0}
          \drawedge(g0,g00){}
          \drawedge(g0,g01){}
          \drawedge(g10,g100){}
          \drawedge(g10,g101){}
          \drawedge(g11,g110){}
          \drawedge(g11,g111){}
        }

        \put(104,44){
          \node(de)(0,0){}
          \node(d0)(10,-16){}
          \node(d1)(-10,-16){}
          \node(d00)(17.5,-28){}
          \node(d01)(2.5,-28){}
          \node(d10)(0,-32){}
          \node(d11)(-20,-32){}
          \node(d100)(7.5,-44){}
          \node(d101)(-7.5,-44){}
          \node(d110)(-12.5,-44){}
          \node(d111)(-27.5,-44){}

          \drawedge(de,d0){0}
          \drawedge[ELside=r](de,d1){1}
          \drawedge(d1,d10){0}
          \drawedge[ELside=r](d1,d11){1}

          \gasset{AHnb=0,dash={0.3 1}0}
          \drawedge(d0,d00){}
          \drawedge(d0,d01){}
          \drawedge(d10,d100){}
          \drawedge(d10,d101){}
          \drawedge(d11,d110){}
          \drawedge(d11,d111){}
        }
      \end{picture}
    \end{center}
    \caption{The graph $G$ for some $L$ containing $11$.}
    \label{fig:bitree.eps}
  \end{figure}
  
  Suppose that linearly bounded graphs of bounded degree were closed
  under edge reversal. It would follow that $H$ obtained from $G$ by
  reversing the edges labeled by $\diese$ is also a linearly bounded
  graph.  As $H$ is linearly bounded, we can assume that $\era{0}{H}$,
  $\era{1}{H}$ and $\era{\diese}{H}$ are incremental context-sensitive
  transductions. Let $x$ be the vertex of $H$ corresponding to
  $\bar{0}$. The set of vertices $F=\textrm{Dom}(\era{\diese}{H})$ is
  a context-sensitive language. It follows from Proposition
  \ref{prop:trace} that $L(H,\{x\},F)$ is context-sensitive. As $L =
  L(H,\{x\},F) \cap \{0,1\}^*$, $L$ would also be context-sensitive
  which contradicts its definition.

  As rational graphs are closed under edge reversal, it follows from
  Thm. \ref{thm:bd-rat-lbg} that rational graphs of bounded degree are
  strictly contained in linearly bounded graphs of bounded degree. \qed
\end{proof}

It may be interesting at this point to recall that there are strong
reasons to believe that the languages accepted by finite degree
synchronized graphs are strictly included in context-sensitive
languages (see \cite{Carayol05}). Furthermore, all existing proofs
that the rational graphs accept the context-sensitive languages break
down when the out-degree is bounded. It is not clear whether rational
graphs of bounded degree accept all context-sensitive languages.
However, as noted in Prop. \ref{prop:trace-det}, it is still the case
for bounded degree linearly bounded graphs, and in particular for
deterministic linearly bounded graphs.

\section{Conclusion}
\label{sect:concl}

This paper gives a natural definition of a class of canonical graphs
associated to the observable computations of labeled linearly bounded
machines. It provides equivalent characterizations of this class as
the Cayley-type graphs of length-decreasing term-rewriting systems,
and as the graphs defined by a sub-class of context-sensitive
transductions which can increase the length of their input by at most
a constant number of letters. Although of a sensibly different nature
from rational graphs, we showed that all rational graphs of bounded
degree are linearly bounded graphs of bounded degree, and that this
inclusion is strict. This leads us to consider a more restricted
notion of infinite automata, closer to classical finite automata (as
was already observed in \cite{Carayol05}), and to propose a hierarchy
of classes of infinite graphs of bounded degree accepting the classes
of languages of the Chomsky hierarchy, in the spirit of
\cite{Caucal02} (see Fig.  \ref{fig:bdh}).

Finite graphs obviously have a bounded degree, they accept rational
languages. The transition graphs of real-time pushdown automata, which
accept all context-free languages, are the regular graphs
\cite{Muller85}, or equivalently the bounded degree HR graphs
\cite{Courcelle89} and bounded degree prefix-recognizable graphs
\cite{Caucal01}. By Prop. \ref{prop:trace-det}, the languages of
deterministic linearly bounded graphs are the context-sensitive
languages.  Deterministic graphs have by definition a bounded degree,
so bounded degree linearly bounded graphs also accept the same class
of languages. Finally, since deterministic Turing machines have
bounded-degree transition graphs and accept all recursively enumerable
languages, we can also restrict the class of Turing graphs to bounded
degree.

\begin{figure}
  \begin{center}
    \begin{picture}(0,61) 

      \drawoval(0,6,64,12,6)
      \drawoval(0,12,67,24,6)
      \drawoval(0,18,70,36,6)
      \drawoval(0,24,73,48,6)
      \drawoval(0,30,76,60,6)

      \gasset{Nframe=n}
      \node(fin)(0,6){
        \begin{tabular}{c}
          Finite graphs
          \\
          \emph{Rational languages}
        \end{tabular}}
      \node(reg)(0,18){
        \begin{tabular}{c}
          Bounded-degree regular graphs
          \\
          \emph{Context-free languages}
        \end{tabular}}
      \node(rat)(0,30){
        \begin{tabular}{c}
          Bounded-degree rational graphs
          \\
          \emph{?}
        \end{tabular}}
      \node(lin)(0,42){
        \begin{tabular}{c}
          Bounded-degree linearly bounded graphs
          \\
          \emph{Context-sensitive languages}
        \end{tabular}}
      \node(tur)(0,54){
        \begin{tabular}{c}
          Bounded-degree Turing graphs
          \\
          \emph{Recursively enumerable languages}
        \end{tabular}}
    \end{picture}
  \end{center}
  \caption{A Chomsky-like hierarchy of bounded-degree infinite
    graphs.}
  \label{fig:bdh}
\end{figure}
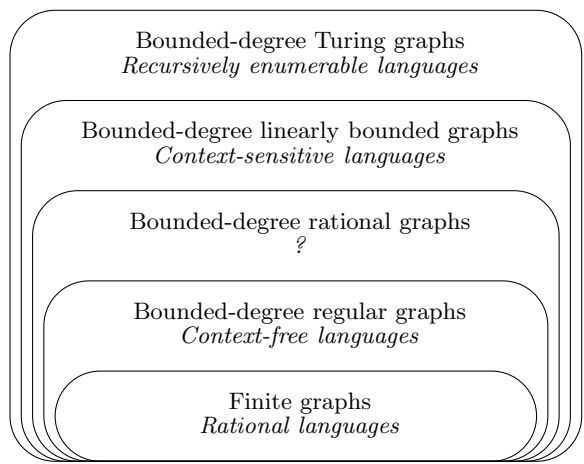  

\bibliographystyle{spmpsci} 
\bibliography{acta-29-05}

\begin{thebibliography}{10}
\providecommand{\url}[1]{{#1}}
\providecommand{\urlprefix}{URL }
\expandafter\ifx\csname urlstyle\endcsname\relax
  \providecommand{\doi}[1]{DOI~\discretionary{}{}{}#1}\else
  \providecommand{\doi}{DOI~\discretionary{}{}{}\begingroup
  \urlstyle{rm}\Url}\fi

\bibitem{Arnold82}
Arnold, A., Nivat, M.: Comportements de processus.
\newblock In: Colloque de l'Association Francaise pour la Cybern{\'e}tique
  {\'E}conomique et Th{\'e}orique: Les Math{\'e}matiques de l'Informatique
  (AFCET~'82), pp. 35--68 (1982)

\bibitem{Blumensath00}
Blumensath, A., Gr{\"a}del, E.: Automatic structures.
\newblock In: Proceedings of the 15th IEEE Symposium on Logic in Computer
  Science (LICS 2000), pp. 51--62. IEEE (2000)

\bibitem{Carayol05}
Carayol, A., Meyer, A.: Context-sensitive languages, rational graphs and
  determinism (2005).
\newblock Submitted

\bibitem{Caucal96}
Caucal, D.: On infinite transition graphs having a decidable monadic theory.
\newblock In: Proceedings of 23rd International Colloquium on Automata,
  Languages, and Programming (ICALP 1996), \emph{Lecture Notes in Computer
  Science}, vol. 1099, pp. 194--205. Springer Verlag (1996)

\bibitem{Caucal03pr}
Caucal, D.: On infinite transition graphs having a decidable monadic theory.
\newblock Theoretical Computer Science \textbf{290}, 79--115 (2003)

\bibitem{Caucal03tm}
Caucal, D.: On the transition graphs of {T}uring machines.
\newblock Theoretical Computer Science \textbf{296}, 195--223 (2003)

\bibitem{Caucal01}
Caucal, D., Knapik, T.: An internal presentation of regular graphs by
  prefix-recognizable graphs.
\newblock Theoretical Computer Science \textbf{34}, 299--336 (2001)

\bibitem{Caucal02}
Caucal, D., Knapik, T.: A {Chomsky}-like hierarchy of infinite graphs.
\newblock In: Proceedings of the 27th International Symposium on Mathematical
  Foundations of Computer Science (MFCS 2002), \emph{Lecture Notes in Computer
  Science}, vol. 2420, pp. 177--187. Springer Verlag (2002)

\bibitem{Chomsky59}
Chomsky, N.: On certain formal properties of grammars.
\newblock Information and Control \textbf{2}, 137--167 (1959)

\bibitem{Courcelle89}
Courcelle, B.: The monadic second-order logic of graphs, {II}: Infinite graphs
  of bounded width.
\newblock Mathematical System Theory \textbf{21}, 187--221 (1989)

\bibitem{Elgot65}
Elgot, C., Mezei, J.: On relations defined by finite automata.
\newblock IBM Journal of Research and Development \textbf{9}, 47--68 (1965)

\bibitem{FrougnyS93}
Frougny, C., Sakarovitch, J.: Synchronized rational relations of finite and
  infinite words.
\newblock Theoretical Computer Science \textbf{108}(1), 45--82 (1993)

\bibitem{Hopcroft79}
Hopcroft, J., Ullman, J.: Introduction to Automata Theory, Languages and
  Computation.
\newblock Addison-Wesley (1979)

\bibitem{Immerman88}
Immerman, N.: Nondeterministic space is closed under complementation.
\newblock SIAM Journal on Computing \textbf{17}(5), 935--938 (1988)

\bibitem{KhoussainovN94}
Khoussainov, B., Nerode, A.: Automatic presentations of structures.
\newblock In: Selected papers of the 1994 International Workshop on Logic and
  Computational Complexity, \emph{Lecture Notes in Computer Science}, vol. 960,
  pp. 367--392. Springer Verlag (1994)

\bibitem{Knapik99}
Knapik, T., Payet, {\'E}.: Synchronized product of linear bounded machines.
\newblock In: Proceedings of the 12th International Symposium on Fundamentals
  of Computation Theory, (FCT 1999), \emph{Lecture Notes in Computer Science},
  vol. 1684, pp. 362--373. Springer Verlag (1999)

\bibitem{Kuroda64}
Kuroda, S.: Classes of languages and linear-bounded automata.
\newblock Information and Control \textbf{7}(2), 207--223 (1964)

\bibitem{Latteux98}
Latteux, M., Simplot, D., Terlutte, A.: Iterated length-preserving rational
  transductions.
\newblock In: Proceedings of the 23rd International Symposium on Mathematical
  Foundations of Computer Science (MFCS 1998), \emph{Lecture Notes in Computer
  Science}, vol. 1450, pp. 286--295. Springer Verlag (1998)

\bibitem{Milner89}
Milner, R.: Communication and Concurrency.
\newblock Prentice Hall International (1989)

\bibitem{Morvan00}
Morvan, C.: On rational graphs.
\newblock In: Proceedings of the 3rd International Conference on Foundations of
  Software Science and Computation Structures (FoSSaCS 2000), \emph{Lecture
  Notes in Computer Science}, vol. 1784, pp. 252--266. Springer Verlag (2000)

\bibitem{Morvan01a}
Morvan, C., Stirling, C.: Rational graphs trace context-sensitive languages.
\newblock In: Proceedings of the 26th International Symposium on Mathematical
  Foundations of Computer Science (MFCS 2001), \emph{Lecture Notes in Computer
  Science}, vol. 2136, pp. 548--559. Springer Verlag (2001)

\bibitem{Muller85}
Muller, D.E., Schupp, P.E.: The theory of ends, pushdown automata, and
  second-order logic.
\newblock Theoretical Computer Science \textbf{37}, 51--75 (1985)

\bibitem{Payet00}
Payet, {\'E}.: Thue specifications, infinite graphs and synchronized product.
\newblock Ph.D. thesis, Universit{\'e} de la R{\'e}union (2000)

\bibitem{Rispal02}
Rispal, C.: The synchronized graphs trace the context-sensitive languages.
\newblock In: Proceedings of the 4th International Workshop on Verification of
  Infinite-State Systems (INFINITY 2002), \emph{Electronic Notes in Theoretical
  Computer Science}, vol.~68 (2002)

\bibitem{Stirling00}
Stirling, C.: Decidability of bisimulation equivalence for pushdown processes.
\newblock Tech. Rep. EDI-INF-RR-0005, School of Informatics, University of
  Edinburgh (2000)

\bibitem{Szelepcsenyi88}
Szelepcsenyi, R.: The method of forced enumeration for nondeterministic
  automata.
\newblock Acta Informatica \textbf{26}, 279--284 (1988)

\bibitem{Thomas01}
Thomas, W.: A short introduction to infinite automata.
\newblock In: Proceedings of the 5th International Conference on Developments
  in Language Theory (DLT 2001), \emph{Lecture Notes in Computer Science}, vol.
  2295, pp. 130--144. Springer Verlag (2001)

\bibitem{Weber96}
Weber, A.: Decomposing a {\it k}-valued transducer into {\it k} unambiguous
  ones.
\newblock {I}nformatique {T}h{\'e}orique et {A}pplications \textbf{30}(5),
  379--413 (1996)

\end{thebibliography}

\end{document}